\documentclass{amsart}
\usepackage[utf8]{inputenc}
\usepackage{amsmath,amssymb,amsthm,tikz}
\usepackage{mathtools}
\usetikzlibrary{decorations.markings,arrows.meta,arrows}
\usepackage{comment}
\usepackage{hyperref}
\usepackage{graphicx}
\usepackage{url}
\usepackage{enumerate}
\usepackage{caption}
\usepackage{mathdots}
\usepackage{float}
 \usepackage{relsize}
 \usepackage{adjustbox}

\newtheorem{theorem}{Theorem}

\newtheorem{lemma}[theorem]{Lemma}
\newtheorem{remark}[theorem]{Remark}
\newtheorem{proposition}[theorem]{Proposition}

\newtheorem{hypothesis}[theorem]{Hypothesis}

\newtheorem{definition}[theorem]{Definition}

\numberwithin{equation}{section}
\numberwithin{theorem}{section}

\DeclareMathOperator{\Sym}{Sym}

\DeclareMathOperator{\Div}{Divisor}

\newcommand{\z}{\#z}
\newcommand{\p}{\#p}
\newcommand{\vs}{\vec{s}}
\newcommand{\vz}{\vec{z}}

\newcommand{\vq}{\vec{q}}
\newcommand{\vr}{\vec{r}}

\newcommand{\vm}{\vec{m}}
\newcommand{\ve}{\vec{e}}
\newcommand{\vo}{\vec{\omega}}
\newcommand{\Ee}{E_{\vec{e}}}
\newcommand{\vte}{\vartheta_{\vec{e}}}
\newcommand{\ff}{\delta_{\be}}
\newcommand{\RR}{\rho_{\be, \boldsymbol{\nu}}}
\newcommand{\bx}{\mathbf{x}}
\newcommand{\be}{\boldsymbol{\eta}}
\newcommand{\bnu}{\boldsymbol{\nu}}

\newcommand{\AC}{\mathcal A \mathcal C}
\newcommand{\GA}{\mathbf G_N^{\text{Aztec}}}

\DeclareMathOperator{\dist}{dist}
\newcommand{\I}{\mathrm{i}}

\newcommand{\R}{\mathbb{R}}
\newcommand{\Z}{\mathbb{Z}}

\newcommand{\C}{\mathbb{C}}

\newcommand{\cR}{\mathcal{R}}

\newcommand{\sfP}{\mathsf{P}}

\makeatletter
\@namedef{subjclassname@1991}{\emph{2020} Mathematics Subject Classification}
\makeatother

\author[M. Piorkowski]{Mateusz Piorkowski}
\address{Department of Mathematics\\KTH Royal Institute of Technology\\
Stockholm 
\\
Sweden}
\email{\href{mailto:mateuszp@kth.se}{mateuszp@kth.se}}

\title[Arctic curves of periodic dimer models]{Arctic curves of periodic dimer models and generalized discriminants}

\date{}

\thanks{This research was
supported by the Methusalem grant METH/21/03 – long term structural funding of
the Flemish Government, and the starting grant from the Ragnar S\"oderbergs Foundation. }
\begin{document}


\begin{abstract}
 We compute the algebraic equation for arctic curves of the Aztec diamond with a doubly (quasi-)periodic weight structure and obtain similar results for certain models of the hexagon. In particular, we determine the algebraic degree of such curves as a function of the number of frozen and smooth (or gaseous) regions. The key to our result is the construction of a discriminant for meromorphic differentials on a higher genus Riemann surface. This construction works analogously for meromorphic sections of arbitrary holomorphic line bundles. In the genus $g = 0$ case this notion reduces to the usual discriminant of a polynomial.
\end{abstract}
\maketitle
\tableofcontents

\section{Introduction}
\subsection{Background}
Dimer models have gained considerable attention in recent years and numerous techniques have been developed to study their statistics in the large size limit. A seminal result due to Cohn, Kenyon and Propp \cite{CKP} states that the height function of fairly general uniform dimer models converges to a deterministic limit, which is called the \emph{limit shape}, that satisfies a variational principle given in terms of an entropy functional. The simplest example of this phenomenon is the arctic circle theorem due to Jockusch, Propp and Shor \cite{JPS} for the Aztec diamond with uniform weights, see Fig.~\ref{Fig:ArcticCircle}.

Much effort has been devoted since then to obtain similar results for nonuniform weights and to compute the local fluctuations of the random height function as the size of the dimer model tends to infinity. Here, a key role is played by the fact that dimer models define a determinantal point process with the correlation kernel given in terms of the inverse Kasteleyn matrix, see \cite{K97}. Thus, questions about typical properties of dimer covers in the large size limit can be stated in terms of asymptotic properties of the inverse Kasteleyn matrix. A similar approach can be formulated for models of the Aztec diamond and the hexagon using a bijection between tilings and families of nonintersecting paths. Here, the Eynard--Mehta Theorem \cite{EM98} gives us the determinantal structure, see also \cite[Sect.~4]{DK21}. A comparison between the nonintersecting paths method and the inverse Kasteleyn matrix method can be found in \cite[Sect.~4]{CD23}.
\begin{figure}[ht]
    \centering
\begin{minipage}{.5\textwidth}
  \centering
  \includegraphics[width=0.92\linewidth]{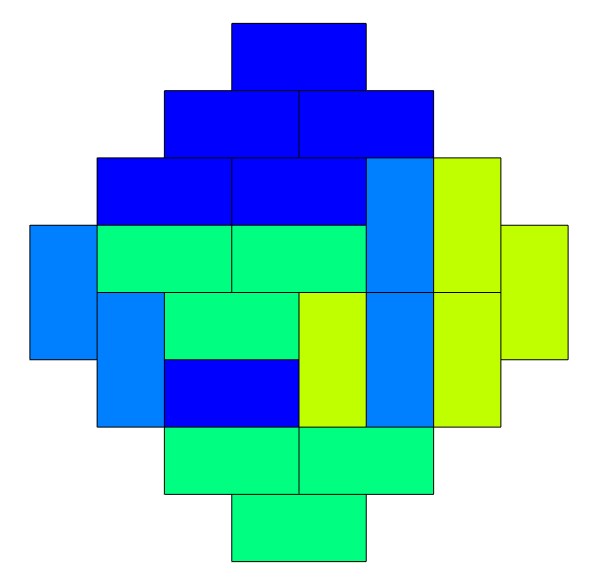}
\end{minipage}%
\begin{minipage}{.5\textwidth}
  \centering
  \includegraphics[width=0.95\linewidth]{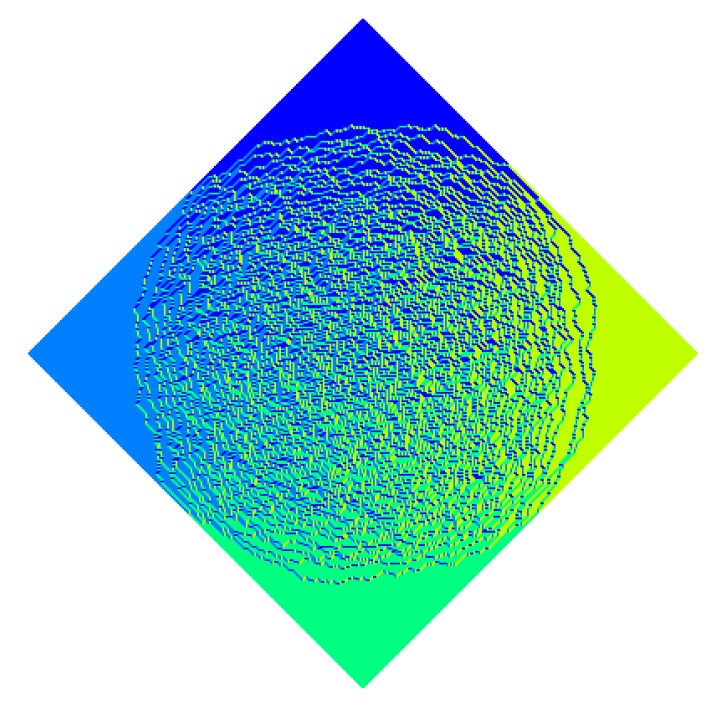}
  
\end{minipage}
\captionof{figure}{A random tiling of the Aztec diamond of size $4$ (left) and of size $200$ (right). Note the frozen regions in the corners and the emergence of a disc-shaped rough region in the middle. In this case the arctic curve becomes a circle as proven by Jockusch, Propp and Shor \cite{JPS} (generated using code kindly provided by Christophe Charlier).}
  \label{Fig:ArcticCircle}
\end{figure}

As shown in the seminal work of Kenyon, Okounkov and Sheffield \cite{KOS06} universal behavior for the height function fluctuations emerge under the assumption that the weights have a doubly periodic structure. In particular, the local dimer statistics converge to a translation invariant Gibbs ensemble of which there are three types: frozen, rough and smooth. In the frozen (or solid) case the location of individual dimers is deterministic and follows some `brick-like' structure; in the rough case the dimer correlations decay polynomially  and the height function displays logarithmic variations; in the smooth case the dimer correlations decay exponentially and the variations of the height function are of order $O(1)$.\footnote{{The papers \cite{ADPZ}, \cite{BB24} make a distinction between frozen and \emph{quasi}-frozen regions. In this terminology, frozen regions consist only of one type of domino/lozenge and correspond to the corners of the Newton polygon, while quasi-frozen regions consist of a deterministic pattern containing \emph{multiple} types of dominos/lozenges and corresponds to the remaining points on the boundary of the Newton polygon. As this distinction is not relevant for our purposes, we will refer collectively to these regions as \emph{frozen regions}.}} 

One of the key insights of \cite{KOS06} is that the space of translation invariant Gibbs measures of a dimer model on the torus is naturally parametrized in terms of the so-called amoeba of an associated Harnack curve, or equivalently its Newton polygon (see \cite[Ch.~6]{GKZ94}). This parametrization leads to a  rich interplay  between algebraic geometry and dimer models with periodic weights, see also \cite{KO06}. In fact, the algebro-geometric description of dimer models remains of central importance also in more recent work, see e.g.~\cite{ADPZ}, \cite{BB23+}, \cite{BD19}, \cite{BB24}, \cite{BBS24}.

While in \cite{KOS06} dimer configurations on the torus were studied, the results therein led to natural conjectures regarding the appearance of the three types of translation invariant Gibbs measure for finite dimer models in the plane where boundary conditions play a role. The typical phase diagram that one observes contains all three phases separated by so-called \emph{arctic curves}, cf.~Fig.~\ref{LargeTiling}.
\begin{figure}[ht]
    \centering
\begin{minipage}{.5\textwidth}
  \centering
  \includegraphics[width=0.95\linewidth]{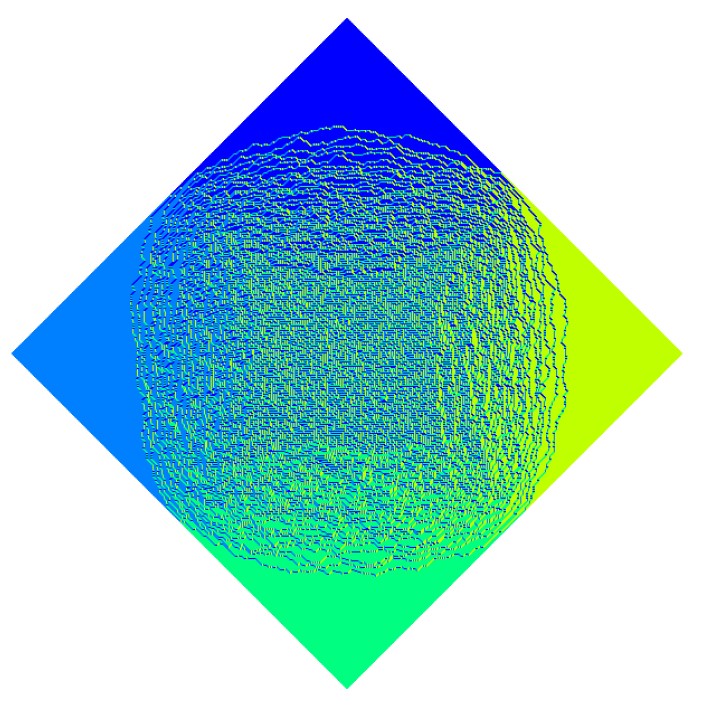}
\end{minipage}%
\begin{minipage}{.5\textwidth}
  \centering
  \includegraphics[width=0.95\linewidth]{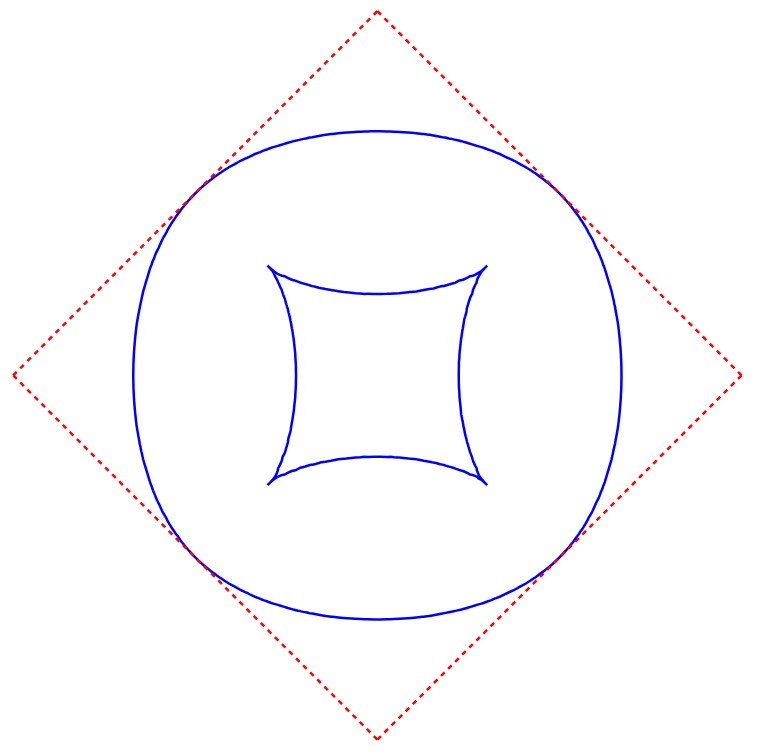}
  
\end{minipage}
\captionof{figure}{\textbf{Left}: Large tiling of the $2 \times 2$-periodic model of the Aztec diamond studied in \cite{CJ16}, \cite{CY14}, \cite{DK21} with one smooth and four frozen regions (size $300$); \textbf{Right}: The arctic curve of degree eight for the corresponding model (left image generated using code kindly provided by Christophe Charlier).}
  \label{LargeTiling}
\end{figure}
As these curves describe the phase transition between different types of Gibbs measures they are of particular interest.  Hence, considerable effort has been devoted to the study of dimer correlations in their vicinity, see e.g.~\cite{BCJ18}, \cite{BCJ22}, \cite{CJ16}, \cite{Jo05}, \cite{JM23}. It turns out that these curves additionally display universal geometric/topological properties that can be in many cases explicitly determined. For a detailed description of the geometry of arctic curves for $k \times \ell$-periodic models of the Aztec diamond see \cite{BB23+}.

\medskip

In the present paper we will be interested in the \emph{algebraic} properties of arctic curves, in particular their degree, which has been computed only for specific cases in a handful of examples, see \cite{BD23}, \cite{CJ16}, \cite{DiFS14}, \cite{DK21} (for some general considerations see \cite{ADPZ}, \cite{KO07}). To the best of the author's knowledge, this is the first time the degree of arctic curves in the case of a dimer model with \emph{generic doubly periodic weights} -- in particular without any restriction on the number of smooth phases -- has been determined. More precisely, for the $k\times \ell$-periodic model of the Aztec diamond studied in \cite{BB23+} the formula for the degree of the underlying arctic can be expressed in terms of the number of smooth and frozen regions as follows:
\begin{align}\label{DegFormula}
    \textbf{degree of arctic curve} = 6\cdot \#\textbf{smooth regions}+2 \cdot \# \textbf{frozen regions} -6.
\end{align}
We prove this result by obtaining an explicit expression in terms of a generalized discriminant for the polynomial whose zero set is the arctic curve, see Fig.~\ref{Fig:2x2}. The strikingly simple formula \eqref{DegFormula} also holds for quasi-periodic models of the Aztec diamond which are defined in terms of Fock weights (see \cite{BB24}, \cite{BdT24}, \cite{Fock}). In fact, the characterization of arctic curves in \cite{BB24} allows us to determine their degree for a quite general class of quasi-periodic models of the hexagon as well. For other dimer models, and the challenges that remain, see Remark \ref{Rmk:Other_Models}.
\begin{figure}[ht]
    \centering
  
  \includegraphics[width=0.7\linewidth]{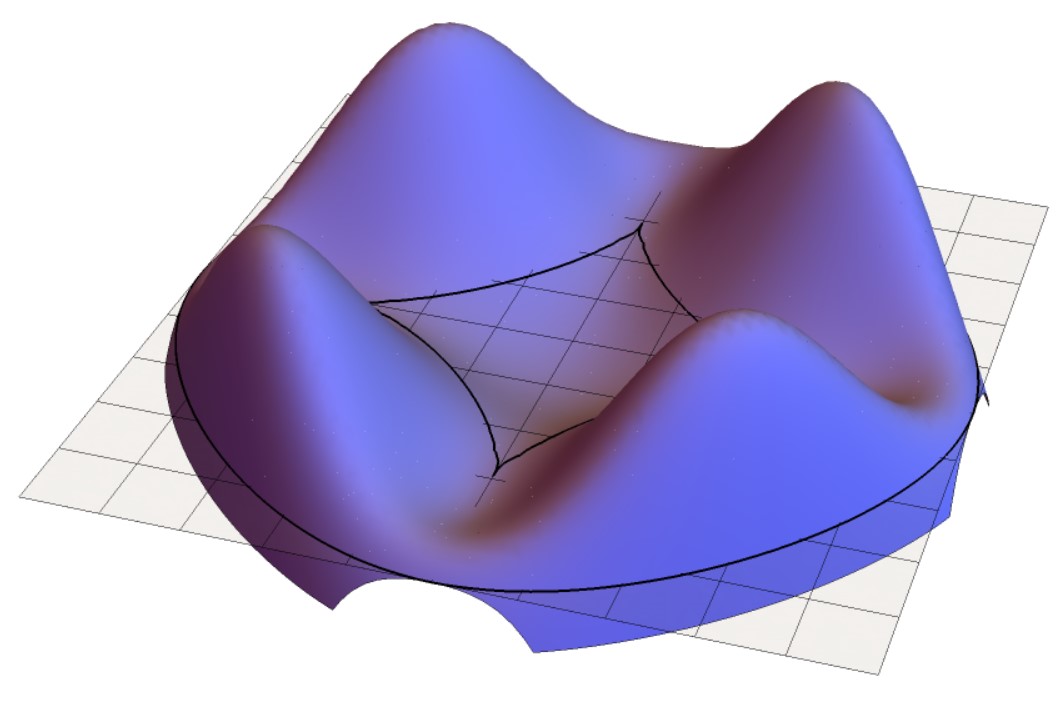}
  
\captionof{figure}{Graph of the $\be$-discriminant $\Delta_{\be}$ of the $2\times2$-periodic model from Fig.~\ref{LargeTiling}. Its zero set highlighted in black coincides with the arctic curve of the model.}
  \label{Fig:2x2}
\end{figure}

\medskip 

The key to our result is the construction of what we will refer to as \emph{$\be$-discriminant}. Here $\be = x_0\eta_0 +  \dots + x_n\eta_n$ is a linear form of meromorphic differentials $\eta_j$ defined on some Riemann surface.\footnote{In the case of the Aztec diamond this surface coincides with the spectral curve, while in the case of the hexagon it is a double cover of the spectral curve.} The particular choice depends on the dimer model we want to study. The $\be$-discriminant can be viewed as an analogue of the $\mathbf A$-discriminant which appeared in the work of  Gelfand, Kapranov and 
Zelevinsky on $\mathbf A$-hypergeometric functions and hyperdeterminants \cite{GKZ90}, \cite{GKZ94}. The $\mathbf A$-discriminant, where $\mathbf A$ denotes a finite set of monomials in $d$ variables, is essentially a generalization of the classical discriminant of a polynomial in one variable to the multi-variable case, i.e.~the domain $\C$ of the polynomial is replaced by $\C^d$ (see \cite[Ch.~9]{GKZ94}). In a similar vein we replace $\C \subset  \C P^1$ with a general compact Riemann surface $\mathcal R$ of arbitrary genus and the underlying polynomials with meromorphic sections of some line bundle. This leads naturally to the definition of $\be$-discriminants (see Def.~\ref{DefDiscr}) and the appearance of Riemann theta functions in their construction.

The importance of these generalized discriminants lies in the fact that for many (quasi-)periodic dimer models, including the Aztec diamond and the hexagon, one can choose $\be$ such that the corresponding $\be$-discriminant gives us the algebraic equation for the arctic curve, see Fig.~\ref{Fig:2x2}. This is our main motivation for defining and studying $\be$-discriminants. In particular, the degree of the arctic curve can be read off easily from an explicit formula for the corresponding $\be$-discriminant as shown in Theorem \ref{Thm:Aztec2}. 

\medskip

While our primary goal is computing arctic curves, the actual concept of an $\be$-discriminant does not require any reference to dimer models. As such, we hope that $\be$-discriminants will also find applications in other areas of mathematics. Readers interested only in generalized discriminants can go straight to the Sections \ref{Sect:Definition} and \ref{Sect:Construction}, which contain the definition, respectively construction, of $\be$-discriminants without referencing dimer models. 

It turns out that only the characterization of arctic curves via a coalescence of zeros condition for meromorphic differentials (see Lem.~\ref{LemmaZeros}) is used in our main result on the Aztec diamond found in Theorem \ref{Thm:Aztec2}. Notions like Harnack curves (or more generally M-curves), the amoeba map and translation invariant Gibbs measures, which are of fundamental importance in the area of dimer models, do not appear in our construction. As such, we choose to present topics related to dimer models in a rather informal way, referring mostly to the existing literature. However, we give a brief description of the $k\times \ell$-periodic Aztec diamond in Section \ref{SectModel}; a more detailed account can be found in \cite[Sect.~2]{BB23+}. For the quasi-periodic models with Fock weights see the recent \cite{BB24}, \cite{BBS24}, \cite{BdT24} (planar dimer models were considered in \cite{BCdT1}, \cite{BCdTg}). For a detailed description of frozen, rough and smooth regions see the work of Kenyon, Okounkov and Sheffield \cite{KOS06}. Important general results on  arctic curves, with emphasis on their algebraic properties, can be found in \cite{ADPZ} and \cite{KO07}. For the concepts of algebraic geometry that are used in the study of dimer models see \cite{GKZ94}, \cite{KOS06}. The general theory of Riemann surfaces and related topics can be found in \cite{GH94}, \cite{Miranda}, \cite{Tata1}.

\subsection{Overview of the paper}
We will now briefly summarize the rest of the paper. 

In Section \ref{SectModel} we give a short description of the $k \times \ell$-periodic Aztec diamond studied in \cite{BB23+}. This model has been analysed in great detail and therefore will serve as our main example. However, we also mention other (quasi-)periodic models of the Aztec diamond and the hexagon in lesser detail, referring instead to existing literature. We then motivate the notion of an $\be$-discriminant that plays a central role in the proof of our results on arctic curves.

In Section \ref{Sect:Definition} we define rigorously  $\be$-discriminants in Definition \ref{DefDiscr}. We do this in a more general setting than required for our application to dimer models.

In Section \ref{Sect:Construction} we start by showing that $\be$-discriminants in the genus $0$ case can be written in terms of the classical discriminant of a polynomial. This part is not new but included for the sake of completeness so that the reader can become familiar with the notion of $\be$-discriminants. We then proceed with the construction of $\be$-discriminants in the higher genus case which involves the use of Riemann theta functions. This construction lies at the heart of the present work. The key result is contained in Theorem \ref{mainTheorem}. 

In Section \ref{Sect:ArcticCurves} we show that $\be$-discriminants define in many cases the algebraic equation for arctic curves. To this end we state in Lemma \ref{LemIrr} a general result on the irreducibility of arctic curves which will imply that their degree is equal to the degree of the associated $\be$-discriminant. We then proceed to go through the main examples for which our result applies. In all cases the degree of the arctic curve turns out to be a simple affine linear function in the number of frozen and smooth regions. 

In Section \ref{Sect:MixRes} we finish the paper by discussing mixed resultants studied in \cite[Ch.~3.3]{GKZ94}, using the methodology developed in the previous sections. This part is not directly related to dimer models but might be of separate interest. In particular, we obtain formulas for these objects in terms of Riemann theta functions. We then show that, as in the classical case, the $\be$-discriminant can be written in terms of a resultant. This gives us an alternative construction of $\be$-discriminants. 
\section{(Quasi-)periodic dimer models}\label{SectModel}
\subsection{\texorpdfstring{$k \times \ell$}{}-periodic Aztec diamond}

We describe now briefly doubly periodic models of the Aztec diamond as studied in \cite{BB23+}. In Section \ref{Sect:ArcticCurves} we will also state our result on arctic curves in the more general setup of \emph{quasi}-periodic weights. The definition of quasi-periodic models of both the Aztec diamond and the hexagon can be found in the recent paper \cite{BB24}.

For some fixed size $N > 0$ of the Aztec diamond  consider the following finite set of points in the plane:
\begin{align*}
    \mathbf B_N &=  \big\lbrace (\tfrac{1}{2}-N+i+j, \, -\tfrac{1}{2}-i+j) \ | \ i = 0, \dots, N-1, \ j = 0, \dots, N \big\rbrace,
    \\
    \mathbf W_N &=  \big\lbrace (\tfrac{1}{2}-N+i+j, \, \tfrac{1}{2}-i+j) \ | \ i = 0, \dots, N, \ j = 0, \dots, N-1 \big\rbrace.
\end{align*}
We refer to $\mathbf B_N$ and $\mathbf W_N$ as the set of black, respectively white, vertices. Define $\mathbf V_N = \mathbf B_N \cup \mathbf W_N$ to be the set of all vertices and consider the graph $\mathbf G_N^{\text{Aztec}} = (\mathbf V_N, \mathbf E_N)$, where $\mathbf E_N$ is the edge-set containing nearest neighbour edges. We refer to $\mathbf G_N^{\text{Aztec}}$ as the \emph{Aztec diamond graph of size $N$}. An example of $\GA$ of size $N = 4$ is shown in Fig.~\ref{ADGraph}.

\begin{figure}[ht]
	\begin{center}
	\begin{tikzpicture}[scale=0.8]
\foreach \j in {0,...,3}{
    \foreach \k in {0,...,4}{
        \node
        (\j, \k) at (0.5-4.0+\j+\k,-0.5-\j+\k) {$\bullet$};
    }
}

\foreach \j in {0,...,4}{
    \foreach \k in {0,...,3}{
        \node
        (\j, \k) at (0.5-4.0+\j+\k,0.5-\j+\k) {$\circ$};
    }
}

\foreach \j in {0,...,3}{
    \foreach \k in {0,...,3}{
        \draw (0.5-4.0+\j+\k,-0.5-\j+\k) -- (0.5-4.0+\j+\k+0.93,-0.5-\j+\k);
        \draw (0.5-4.0+\j+\k,-0.5-\j+\k) -- (0.5-4.0+\j+\k,-0.5-\j+\k+0.93);
    }
}

\foreach \j in {0,...,3}{
    \foreach \k in {1,...,4}{
        \draw (0.5-4.0+\j+\k,-0.5-\j+\k) -- (0.5-4.0+\j+\k-0.93,-0.5-\j+\k);
        \draw (0.5-4.0+\j+\k,-0.5-\j+\k) -- (0.5-4.0+\j+\k,-0.5-\j+\k-0.93);
    }
}

\end{tikzpicture}
	\end{center}
	\caption{\label{ADGraph}Aztec diamond graph of size 4.}
\end{figure}
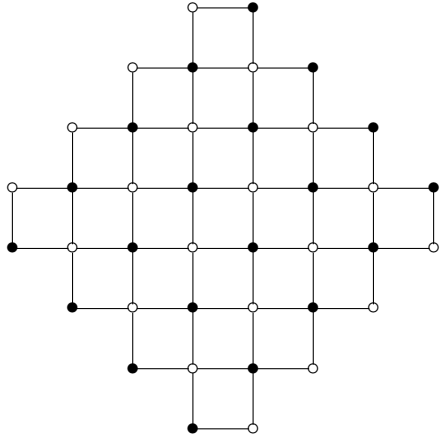
Tilings of the Aztec diamond are in one-to-one correspondence with dimer covers (or perfect matchings) of the graph $\GA$ via the identification between edges and dominos displayed in Fig.~\ref{DimerDominos}.
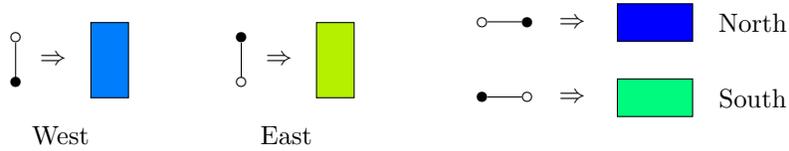
\begin{figure}[ht]
    \centering
    \begin{tikzpicture}
        \node at (0, 0.2)  (a){$\bullet$};
        \node at (0, 0.8) (b) {$\circ$};
        \draw (0,0.2)--(0,0.73);
        \node at (0.5, 0.5)  (a){$\Rightarrow$};
        \draw [fill={rgb,255:red,0; green,125; blue,250}] (1,0) rectangle (1.5,1);
        \node at (0.6, -0.5)  {West};
 
 \node at (3, 0.2)  (a){$\circ$};
        \node at (3, 0.8) (b) {$\bullet$};
        \draw (3,0.27)--(3,0.73);
        \node at (3.5, 0.5)  (a){$\Rightarrow$};
        \draw [fill={rgb,255:red,180; green,240; blue,0}] (4,0) rectangle (4.5,1);
        \node at (3.6, -0.5)  {East};
\node at (6.2, 1)  (a){$\circ$};
        \node at (6.8, 1) (b) {$\bullet$};
        \draw (6.27,1.01)--(6.73,1.01);
        \node at (7.4, 1)  (a){$\Rightarrow$};
        \draw [fill=blue] (8,1.25) rectangle (9,0.75);
\node at (9.8, 1)  {North};
\node at (6.2, 0)  (a){$\bullet$};
        \node at (6.8, 0) (b) {$\circ$};
        \draw (6.27,0.01)--(6.73,0.01);
        \node at (7.4, 0)  (a){$\Rightarrow$};
        \draw [fill={rgb,255:red,0; green,250; blue,125}] (8,-0.25) rectangle (9,0.25);
 \node at (9.8, 0)  {South};
    \end{tikzpicture}
    \caption{\label{DimerDominos}Map between dimers and dominos. The use of similar colors is necessary for the smooth region to be distinguishable from the rough region.}
    
\end{figure}

A $k \times \ell$-periodic model of the Aztec diamond is defined in terms of positive edge-weights $\alpha_{j,i}, \, \beta_{j,i}, \, \gamma_{j,i} > 0$, with $j = 1, \dots, k$ and $i = 1, \dots, \ell$. It is convenient to extend these weights in an $k \times \ell$-periodic fashion by setting $\alpha_{j+mk, i + n\ell} = \alpha_{j,i}$ ect., for $m,n \in \Z$. We also assume that
\begin{equation} \label{polezero} 
	\beta_i^v < 1 < \frac{\alpha_i^v}{\gamma_i^v} \quad \text{for } i =1, \dots, \ell,
\end{equation}
where $\alpha^v_i = \prod_{j=1}^k \alpha_{j,i}$, $\beta^v_i = \prod_{j=1}^k \beta_{j,i}$
and $\gamma^v_i = \prod_{j=1}^k \gamma_{j,i}$. Condition \eqref{polezero} is imposed in \cite[Assumption 4.1(b)]{BB23+} to guarantee the existence of the Wiener--Hopf factorizations. However, it can be dropped in the more general case of quasi-periodic weights, see Section \ref{Sect:ArcAzt}.

The weights are then assigned to the edges of $\GA$ (after a rotation by $45^{\circ}$) in a doubly periodic fashion as indicated in Fig.~\ref{DimerGraph}.
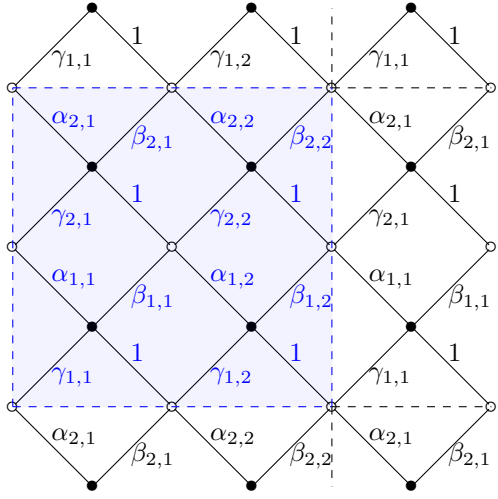
\begin{figure}[ht]
\vspace{-25pt}
\begin{tikzpicture}
    \put(-77,-57){$\gamma_{1,1}$}
    \put(-77,-117){\textcolor{blue}{$\gamma_{2,1}$}}
    \put(-77,-177){\textcolor{blue}{$\gamma_{1,1}$}}

    \put(-77,-80){\textcolor{blue}{$\alpha_{2,1}$}}
    \put(-77,-140){\textcolor{blue}{$\alpha_{1,1}$}}
    \put(-77,-200){$\alpha_{2,1}$}

    \put(-47,-50){$1$}
    \put(-47,-110){\textcolor{blue}{$1$}}
    \put(-47,-170){\textcolor{blue}{$1$}}

    \put(-47,-88){\textcolor{blue}{$\beta_{2,1}$}}
    \put(-47,-148){\textcolor{blue}{$\beta_{1,1}$}}
    \put(-47,-208){$\beta_{2,1}$}


    \put(-17,-57){$\gamma_{1,2}$}
    \put(-17,-117){\textcolor{blue}{$\gamma_{2,2}$}}
    \put(-17,-177){\textcolor{blue}{$\gamma_{1,2}$}}

    \put(-17,-80){\textcolor{blue}{$\alpha_{2,2}$}}
    \put(-17,-140){\textcolor{blue}{$\alpha_{1,2}$)}}
    \put(-17,-200){$\alpha_{2,2}$}

    \put(13,-50){$1$}
    \put(13,-110){\textcolor{blue}{$1$}}
    \put(13,-170){\textcolor{blue}{$1$}}

    \put(13,-88){\textcolor{blue}{$\beta_{2,2}$}}
    \put(13,-148){\textcolor{blue}{$\beta_{1,2}$}}
    \put(13,-208){$\beta_{2,2}$}


    \put(43,-57){$\gamma_{1,1}$}
    \put(43,-117){$\gamma_{2,1}$}
    \put(43,-177){$\gamma_{1,1}$}

    \put(43,-80){$\alpha_{2,1}$}
    \put(43,-140){$\alpha_{1,1}$}
    \put(43,-200){$\alpha_{2,1}$}

    \put(73,-50){$1$}
    \put(73,-110){$1$}
    \put(73,-170){$1$}

    \put(73,-88){$\beta_{2,1}$}
    \put(73,-148){$\beta_{1,1}$}
    \put(73,-208){$\beta_{2,1}$}

\end{tikzpicture}
\vspace{-35pt}
\hspace*{-8pt}
   	\begin{center}
	\rotatebox{45}{
 \begin{tikzpicture}[scale=1.5]
\foreach \j in {1,...,3}{
    \foreach \k in {0,...,3}{
        \node
        (\j, \k) at (0.5-4.0+\j+\k,-0.5-\j+\k) {$\bullet$};
    }
}

\foreach \j in {1,...,4}{
    \foreach \k in {0,...,2}{
        \node
        (\j, \k) at (0.5-4.0+\j+\k,0.5-\j+\k) {$\circ$};
    }
}

\foreach \j in {1,...,3}{
    \foreach \k in {0,...,2}{
       \draw (1.5-4.0+\j+\k,-0.47-\j+\k)--(1.5-4.0+\j+\k,0.5-\j+\k);
        \draw (1.5-4.03+\j+\k,-0.5-\j+\k)--(1.5-5+\j+\k,-0.5-\j+\k);
    }
}

\foreach \j in {0,...,2}{
    \foreach \k in {0,...,2}{
       \draw (1.5-4.0+\j+\k,-0.53-\j+\k)--(1.5-4.0+\j+\k,-1.5-\j+\k);
        \draw (1.5-3.97+\j+\k,-0.5-\j+\k)--(1.5-3+\j+\k,-0.5-\j+\k);
    }
}
\draw [blue, dashed] (-2.5,-0.5) -- (-0.5,1.5);
\draw [blue, dashed] (-0.5,-2.5) -- (1.5,-0.5);

\draw [blue, dashed] (-2.5,-0.5) -- (-0.5,-2.5);
\draw [dashed] 
(-0.5,-2.5) -- (0.5, -3.5);
\draw [dashed]
(-0.5,-2.5) -- (-1, -3);

\draw [blue, dashed] (-0.5,1.5) -- (1.5,-0.5);
\draw [dashed] 
(1.5,-0.5)--(2.5, -1.5);
\draw [dashed]
(1.5,-0.5) -- (2, 0);

\draw [fill=blue, opacity=0.05] (-2.5,-0.5) -- (-0.5,1.5) -- (1.5,-0.5) -- (-0.5,-2.5) -- cycle;
\end{tikzpicture}}
	\end{center}
 \vspace{-70pt}
\caption{Aztec diamond graph of size $N = 3$ with doubly periodic weights having periodicity $k = \ell = 2$. The fundamental domain is depicted in blue, and the weight structure repeats throughout the graph in a doubly periodic manner. \label{DimerGraph}}
\end{figure}
Let $\mathcal T_N$ be the set of all dimer covers of $\GA$. We denote elements of $\mathcal T_N$ by $T = \lbrace e_1, e_2, \dots \rbrace \subset \mathbf E_N$, where $e_n$ are the edges in the dimer cover $T$. Define on $\mathcal T_N$ the following probability measure:
\begin{align*}\text{Prob}_N(T) = \frac{\prod_{e \in T} w_e}{\sum_{T'\in \mathcal T_N}\prod_{e' \in T'} w_{e'}}.
\end{align*}
Here $w_e \in \lbrace \alpha_{j,i}, \, \beta_{j,i}, \, \gamma_{j,i} \, \colon \, j = 1, \dots, k; \, i = 1, \dots, \ell \rbrace$ is the weight of the edge $e \in \mathbf E_N$ (cf.~Fig.~\ref{DimerGraph}).

The large $N$ asymptotics of $k \times \ell$-periodic models of the Aztec diamond have been studied extensively in \cite{BB23+}. For such problems it customary to scale the (rotated) graph $\GA$ by $1/N$ such that in the $N \to \infty$ limit $\GA$ would fill out the unit square $[-1,1]^2$ (cf.~Fig.~\ref{Fig:kxl}). Asymptotic results concerning dimer models, e.g.~conformal structure of the rough region, limit shapes of the height function, geometry of arctic curves ect., can be conveniently expressed in terms of the macroscopic coordinates $(x_1, x_2) \in [-1,1]^2$.

A key result of \cite{BB23+} is that the geometry of the Aztec diamond can be described via the location of zeros of a meromorphic differential $dF=\eta_0 + x_1 \eta_1 + x_2 \eta_2$ depending affine linearly on the global coordinates $(x_1, x_2) \in [-1,1]^2$. The meromorphic differentials $\eta_0, \eta_1, \eta_2$ together with the underlying Riemann surface $\mathcal R$ can be constructed explicitly in terms of the model parameters $\alpha_{j,i}, \beta_{j,i}, \gamma_{j,i}$ as we summarize in the following. First consider the transfer matrices  $\phi^b$, $\phi^g$ given by
\begin{align*}  \phi^b_i(z; \vec{\alpha}_i, \vec{\gamma}_i)
	= \begin{pmatrix} \gamma_{1,i} & 0 & \cdots & 0 & \alpha_{k,i} z^{-1} \\
		\alpha_{1,i} & \gamma_{2,i} & \cdots & 0 & 0  \\
		\vdots & \vdots & \ddots & \vdots & \vdots \\
		0 & 0 & \cdots & \gamma_{k-1,i} & 0 \\
		0 & 0 & \cdots & \alpha_{k-1,i} & \gamma_{k,i} \end{pmatrix} ,\quad \vec{\alpha}_i = \begin{pmatrix}
		    \alpha_{1,i}
      \\
      \vdots
      \\
      \alpha_{k,i}
		\end{pmatrix}, \ \vec{\gamma}_i = \begin{pmatrix}
		    \gamma_{1,i}
      \\
      \vdots
      \\
      \gamma_{k,i}
		\end{pmatrix} \end{align*}
and
\begin{align*}  \phi^g(z; \vec{\beta}_i) = \frac{1}{1- \big(\prod_{j=1}^k \beta_{j,i}\big) z^{-1}} \begin{pmatrix}
		1 & \big(\prod_{j=2}^k \beta_{j,i}\big) z^{-1} & \cdots & \beta_{k,i} z^{-1} \\
		\beta_{1,i} & 1 & \cdots & \beta_{k,i} \beta_{1,i} z^{-1} \\
		\vdots & \vdots & \ddots & \vdots \\
		\prod_{j=1}^{k-1} \beta_{j,i} & \prod_{j=2}^{k-1} \beta_{j,i} & \cdots & 1 \end{pmatrix}, 
\end{align*}
with $\vec{\beta}_i = \begin{pmatrix}
		    \beta_{1,i} & \dots & \beta_{k,i}
      \end{pmatrix}^T$. Here the superscripts $b$ and $g$, stand for \emph{Bernoulli} and \emph{geometric} respectively. These matrices play a fundamental role in double-integral representations of the underlying dimer point process via the Eynard--Mehta Theorem. The reader is referred to \cite[Sect.~2]{BB23+} or \cite[Sect.~4]{DK21} for more details.

      Next consider the product $\Phi(z)= \prod_{i = 1}^\ell \big(\phi^b(z; \vec{\beta}_i) \phi^g(z; \vec{\alpha}_i, \vec{\gamma}_i)\big)$. According to \cite[Prop.~1.2]{BB23+} (see also \cite[Thm.~5.1]{KOS06}) the equation 
      \begin{align}\label{PDet}
     \prod_{i=1}^\ell (1-\beta_i^v z^{-1})\det(\Phi(z)-w I) = 0
      \end{align}
defines a Harnack curve, which is intimately related to tilings of the underlying Aztec diamond. We denote by $\cR$ the compact Riemann surface defined by \eqref{PDet}, which will be by the Harnack property an M-curve with involution $(z,w) \to (\overline{z}, \overline{w})$. Hence, assuming $\cR$ is of genus $g$, there will be $g+1$ real ovals $X_0, \dots, X_g$ containing all \emph{real} solutions $(z,w)$ of \eqref{PDet}. We choose the ovals $X_1, \dots, X_g$ to coincide with the $\mathbf{a}$-cycles of $\cR$, see Fig.~\ref{Fig:Riemann_Surface}.

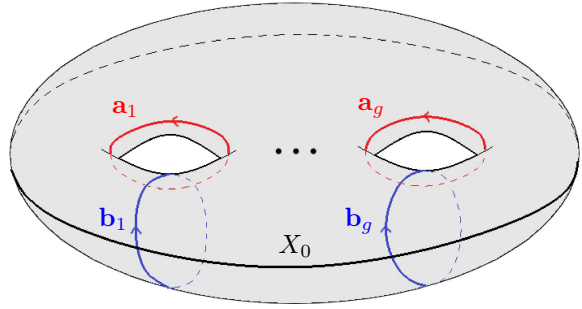
\begin{figure}
    \centering
    \includegraphics[width=0.6\linewidth]{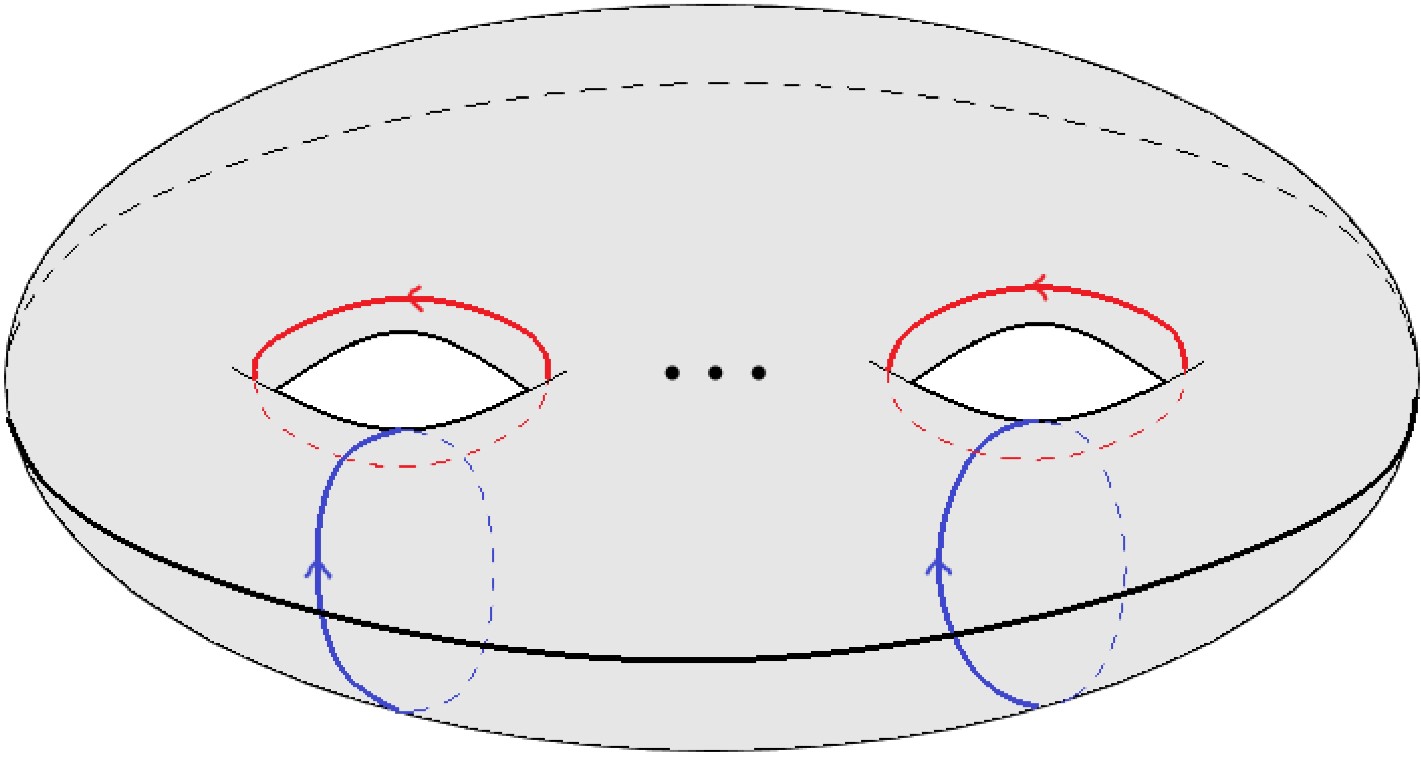}
\begin{tikzpicture}
\put(-180,73){\textcolor{red}{$\mathbf{a}_1$}};
\put(-87,74){\textcolor{red}{$\mathbf{a}_g$}};

\put(-185,30){\textcolor{blue}{$\mathbf{b}_1$}};
\put(-92,30){\textcolor{blue}{$\mathbf{b}_g$}};

\put(-117,20){$X_0$};
    \end{tikzpicture}
    \caption{The Riemann surface $\mathcal R$ with the $\mathbf{a}$-cycles in red corresponding to the real ovals $X_1, \dots, X_g$. The real oval $X_0$ is drawn in black.}
    \label{Fig:Riemann_Surface}
\end{figure}

It turns out that one of the real ovals, often referred to as the \emph{unbounded} oval, comes equipped with a special set of points. Let
\begin{align*}
    q_{0,i} = ((-1)^k \alpha_i^v/\gamma_i^v, 0), \quad q_{\infty,i} = (\beta_i^v, \infty), \quad i = 1, \dots, \ell
\end{align*}
and 
\begin{align*}
    p_{0,j} = (0, (-1)^\ell \alpha_j^h/\beta_j^h), \quad p_{\infty,j} = (\infty, \gamma_j^h), \quad j = 1, \dots, k,
\end{align*}
where $\alpha^v_i = \prod_{j=1}^k \alpha_{j,i}$, $\alpha^h_j = \prod_{i=1}^\ell \alpha_{j,i}$ ect. Then the points $q_{0,i}$, $q_{\infty, i}$, $p_{0,j}$, $p_{\infty,j}$, also called \emph{angles}, all lie on one of the real ovals, say $X_0$, and satisfy a particular order when `going around' $X_0$, see \cite[Sect.~3.2]{BB23+} and Fig.~\ref{Fig:UnboundedOval}.
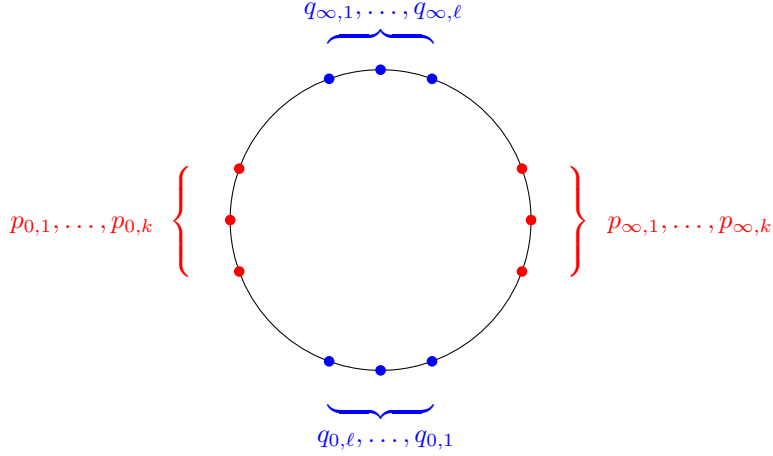
\begin{figure}
\vspace{23pt}
    \begin{tikzpicture}
  \coordinate (center) at (0,2);
  \def\radius{2cm}
  \draw (center) circle[radius=\radius];

  \fill[red] (center) ++(0:\radius) circle[radius=2pt];
  \fill[red] (center) ++(-20:\radius) circle[radius=2pt];
  \fill[red] (center) ++(20:\radius) circle[radius=2pt];
  \fill[red] (center) ++(180:\radius) circle[radius=2pt];
  \fill[red] (center) ++(200:\radius) circle[radius=2pt];
  \fill[red] (center) ++(160:\radius) circle[radius=2pt];


  \fill[blue] (center) ++(90:\radius) circle[radius=2pt];
  \fill[blue] (center) ++(70:\radius) circle[radius=2pt];
  \fill[blue] (center) ++(110:\radius) circle[radius=2pt];
  \fill[blue] (center) ++(-90:\radius) circle[radius=2pt];
  \fill[blue] (center) ++(-70:\radius) circle[radius=2pt];
  \fill[blue] (center) ++(-110:\radius) circle[radius=2pt];

\put(-20,-12){\textcolor{blue}{$\underbrace{\hspace{40pt}}_{}$}};
  \put(-24,-27){\textcolor{blue}{$q_{0, \ell}, \dots, q_{0, 1}$}};

  \put(-20,123){\textcolor{blue}{$\overbrace{\hspace{40pt}}_{}$}};
  \put(-29,135){\textcolor{blue}{$q_{\infty, 1}, \dots, q_{\infty, \ell}$}};

  \put(-80,54){\textcolor{red}{$\begin{cases}
      \\
      \\
      \\
  \end{cases}$}};

  \put(69,54){\textcolor{red}{$\begin{rcases*}
\\
\\
\\
\end{rcases*}$}};

\put(-140,54){\textcolor{red}{$p_{0,1}, \dots, p_{0,k}$}};

\put(86,54){\textcolor{red}{$p_{\infty,1}, \dots, p_{\infty,k}$}};
\end{tikzpicture}
\vspace{29pt}
    \caption{Schematic view of the placement of the angles on the unbounded oval $X_0$.}
    
    \label{Fig:UnboundedOval}
\end{figure}
In the nondegenerate case assumed in \cite{BB23+}, all $2(k+\ell)$ angles are taken to be pairwise disjoint.

Important for us is a certain special meromorphic differential $dF = \eta_0 + x_1 \eta_1 + x_2 \eta_2$ which depends affine linearly on the coordinates $(x_1, x_2) \in [-1,1]^2$ and characterizes the arctic curve. For $\eta_{1,2}$ we have to choose 
\begin{align*}
    \eta_1 = \frac{\ell}{2}\frac{dz}{z}, \qquad \eta_2 = -\frac{k}{2}\frac{dw}{w}.
\end{align*} 
Note that both differentials are real on the real ovals. The remaining meromorphic differential $\eta_0$ is defined as
\begin{align*}
    \eta_0 = -\eta_1-\eta_2 - \frac{df}{f},
\end{align*}
where $\frac{df}{f}$ denotes the unique meromorphic differential with zero periods along the $\mathbf{a}$-cycles and simple poles at $q_{0,i}$ of residue $k$ and at $p_{j,0}$ of residue $-\ell$. It can be explicitly constructed in terms of Fay's prime form on $\cR$, see \cite[Eq.~(1.5)]{BB23+}, \cite{Fay}.

Of central importance is the following property of $dF$ established in \cite{BB23+}:
\begin{lemma}{\emph{(simplified version, cf}.~\cite[Def.~4.7]{BB23+}\emph{)}}\label{LemmaZeros}
    A point $(x_1, x_2) \in [-1,1]^2$ lies on the arctic curve if and only if the meromorphic differential $dF$ has a zero of higher multiplicity. 
\end{lemma}
\begin{remark}\label{Rmk:Other_Models}
The main obstacle in obtaining algebraic equations for arctic curves in a more general class of dimer models (e.g.~the ones considered in \cite{ADPZ}) is precisely the lack of Lemma \ref{LemmaZeros} in these cases. While we know that the above characterization holds  for the general (quasi-)periodic Aztec diamond and certain (quasi-)periodic models of the hexagon, it is presently not clear how to obtain it in general. However, if an analogue of Lemma \ref{LemmaZeros} can be found for some larger class of models, then the equation for the corresponding arctic curves can likely be obtained just as in Section \ref{Sect:ArcAzt}.
\end{remark}
The reader should note that Lemma \ref{LemmaZeros} is a very simplified version of the characterization of the arctic curve in terms of the zeros of $dF$, however completely sufficient for our purposes. Before we continue let us say a few words about the role of $dF$ in the paper \cite{BB23+}. The zeros of $dF$ are the saddle points of a multi-valued phase function $F$ on the Riemann surface $\mathcal R$ that appears in the steepest descent analysis of the correlation kernel for the underlying dimer point process. As such, the location of zeros of $dF$ as a function of $(x_1, x_2) \in [-1,1]^2$ parametrizes the local translation invariant Gibbs measure at that point. A double zero of $dF$ then corresponds exactly to a transition between the rough and frozen/smooth regions (i.e.~a transition in the correlation decay rate of the Gibbs measure), explaining the characterization of arctic curves found in Lemma \ref{LemmaZeros}. The type of transition is again determined by the location of the double zeros; the reader is referred to \cite[Def.~4.7]{BB23+} for more details. Similar descriptions of dimer models in terms of the location of saddle points/zeros of meromorphic differentials have appeared in \cite{B21}, \cite{BD23}, \cite{CDKL20}, \cite{DK21}, and more recently in \cite{BB24} for quasi-periodic models.

We emphasize that the phase function $F$ (and hence the meromorphic differential $dF$) is explicitly given in terms of the parameters of the $k \times \ell$-periodic model treated in \cite{BB23+}. However, important to us are only the following two facts that hold in the nondegenerate case:
\begin{itemize}
    \item the Riemann surface $\mathcal R$ has genus $g = (k-1)(\ell-1)$;
    \item the meromorphic differential $dF$ has $2(k+\ell)$  simple poles.
\end{itemize}
Conversely, a $k \times \ell$-periodic model is nondegenerate if and only if the above conditions hold (for generic weights the model will be nondegenerate). As we will see in Section \ref{Sect:Construction}, this is all the information required to determine the degree of the arctic curve.

In the following sections we will denote the number of poles of $\be$ ($= dF$ in the dimer case) by $\p$. Similarly we will denote the number of zeros of $\be$ by $\z = \# p + 2g -2$. Geometrically $\p$ is the number of frozen regions and the genus $g$ is the number of smooth regions. As a consequence of Lemma \ref{LemmaZeros} and Theorem \ref{mainTheorem} we obtain:
\begin{theorem}\label{Thm:Aztec}
Consider a nondegenerate $k \times \ell$-periodic model of the Aztec diamond and let $\mathcal A\mathcal C$ be its arctic curve. Then $\AC$ is the real part of an algebraic curve of degree $6k\ell - 2(k+\ell) = 6g + 2 \cdot \# p -6$. In particular, formula \ref{DegFormula} holds. 
\end{theorem}
\begin{proof}
    See Theorem \ref{Thm:Aztec2}. 
\end{proof}
\subsection{Other dimer models}
As we mentioned above, Theorem \ref{Thm:Aztec} relies on the recent work of Berggren and Borodin \cite{BB23+} characterizing arctic curves of $k \times \ell$-periodic models of Aztec diamonds in terms of a coalescence of zeros condition for a meromorphic differential depending affine linearly on the coordinates $x_1, x_2$ (see Lem.~\ref{LemmaZeros}). It is however expected, and oftentimes proven, that such a characterization also holds for more general dimer models with (quasi-)periodic weights. Recently, the general case of quasi-periodic weighted models of the Aztec diamond and some quasi-periodic models of the hexagon have been studied in \cite{BB24} and \cite{BBS24}. Crucially, conditions analogous to Lemma \ref{LemmaZeros} have been established, and we provide the corresponding results on arctic curves, together with a short description of the models, in Section \ref{Sect:ArcticCurves}.

\subsection{Notion of a discriminant on a compact Riemann surface}
The preceding discussion naturally leads to the notion of a discriminant on a compact Riemann surface. Consider a domain $\Omega \subset \R^2$ which coincides with the scaled large size limit of a dimer model, e.g.~in the case of the Aztec diamond we have $\Omega = [-1,1]^2$, see Fig.~\ref{Fig:kxl}. 
\begin{figure}[ht]
    \centering
    
  \includegraphics[width=0.5\linewidth]{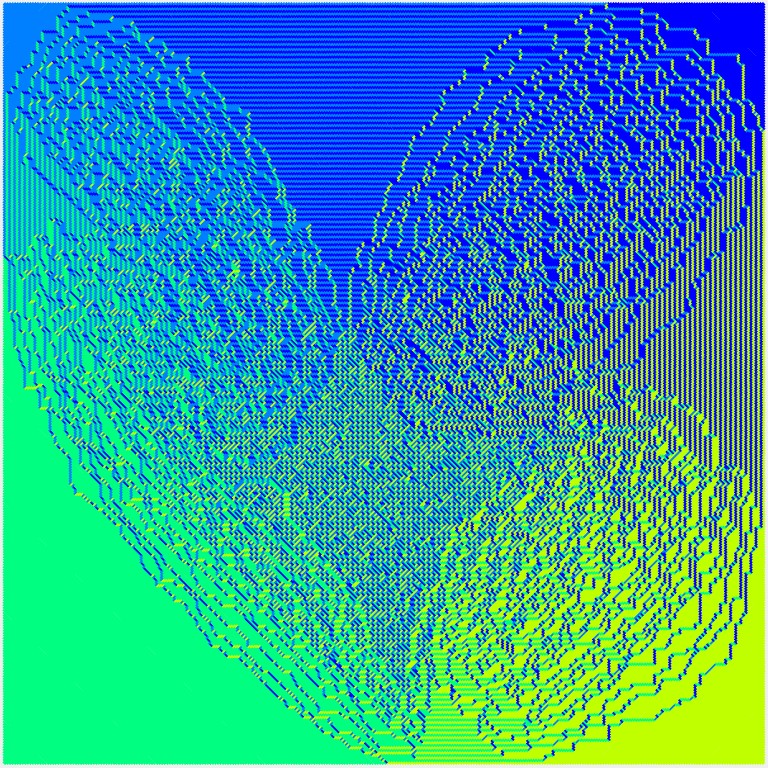}

\captionof{figure}{A random tiling of the rotated Aztec diamond with nondegenerate $2\times2$-periodic weights. There is one smooth region in the middle and eight frozen regions, hence according to \eqref{DegFormula} the arctic curve is of degree $6 \times 1 + 2 \times 8 - 6= 16$. Note that the tiling fills out approximately a square region $\Omega$ (image generated using code kindly provided by Christophe Charlier).}
  \label{Fig:kxl}
\end{figure}
Analogously to \cite{BB23+} we associate to each $(x_1, x_2) \in \Omega$ a meromorphic differential of the form
\begin{align*}
    \be(x_1, x_2) = \eta_0 + x_1\eta_1 + x_2 \eta_2, \qquad (x_1, x_2) \in \Omega,
\end{align*}
where $\eta_j$, $j =0,1,2$ are fixed meromorphic differentials on some compact Riemann surface $\mathcal R$ of genus $g$ (here the choice of $\eta_j$ and $\cR$ depends on the dimer model). We denote the arctic curve that appears in the scaling limit by $\AC \subset \Omega$. Furthermore, we assume the analog of Lemma \ref{LemmaZeros} which we now turn into a hypothesis.
\begin{hypothesis}
    A point $(x_1, x_2) \in \Omega$ lies on the arctic curve $\AC$ if and only if the meromorphic differential $\be$ has a zero of higher multiplicity.
\end{hypothesis}
Thus, it becomes clear that obtaining an algebraic equation for the arctic curve is equivalent to finding a `discriminant-type expression' for $\be(x_1, x_2)$, that is a polynomial in $x_1, x_2$ which vanishes if and only if $\be(x_1, x_2)$ has zeros of higher multiplicity. We will refer to this object as the $\be$-discriminant.  

As we show in Section \ref{SectGenus0}, if the genus of $\mathcal R$ is zero (i.e.~no smooth regions are present), then the $\be$-discriminant can be found easily in terms of the usual discriminant of a polynomial. As a reminder, for a given polynomial $\mathcal P(z) = \sum_{k=0}^n c_k z^k$ with $c_n \not = 0$, its discriminant $\Delta_n(\mathcal P)$ is given by 
\begin{align}\label{Def:Delta_n}
    \Delta_n(\mathcal P) = (-1)^{\frac{n(n-1)}{2}} c_n^{2n-2} \prod_{i \not = j} (z_i - z_j), \qquad z_1, \dots, z_n \text{ are zeros of } P. 
\end{align}
It turns out that for a fixed degree $n$, the discriminant $\Delta_n$ is a homogeneous polynomial of degree $2n-2$ in the coefficients $c_0, \dots c_n$. The main difficulty in generalizing this notion to the higher genus case is the factor $(z_i-z_j)$ in \eqref{Def:Delta_n}, which does not make sense if $z_i, z_j$ are points on an arbitrary compact Riemann surface.

\medskip

We end this section with the following remark explaining why a common approach to essentially circumvent the $\be$-discriminant (by using the classical discriminant instead) becomes increasingly challenging for more complicated dimer models.
\begin{remark}
     As the Riemann surface $\mathcal R$ is given as a Harnack curve via \eqref{PDet}, it comes naturally equipped with a sheet structure. In particular, the resulting projection $\pi \colon \cR \to \mathbb C$ can be used in simple models to determine the algebraic equation for the arctic curve with the help of the \emph{classical} discriminant of a polynomial, see e.g.~\cite[Sect.~A.3]{BD23}, \cite[Sect.~6.3]{DK21}. This is achieved by projecting the zeros of $dF$ via $\pi$ to $\mathbb C$ and taking the discriminant of the polynomial which vanishes at the projected zeros. This discriminant will vanish whenever two zeros coalesce on the original Riemann surface, but also when two distinct zeros are mapped via $\pi$ to the same point in $\C$. Hence, in general the discriminant will have a nontrivial factorization, and only one of the factors defines the arctic curve.

     This approach, if successful, provides an explicit formula for the algebraic equation of the arctic curve in terms of the model parameters. The main drawback is that already for simple models the explicit formulas for the discriminants become very large and one has to rely on computer software to factorize them \cite[Sect.~A.3]{BD23} (explicit equations for arctic curves of models with multiple smooth phases can be found in \cite[Appendix]{DiFS14}). Hence, it is not clear how to extend this method to the general $k \times \ell$-periodic case. In fact, one of the main motivations to introduce $\be$-discriminants is precisely to avoid the aformentioned factorization by working directly on the Riemann surface $\mathcal R$ instead of $\C$.
 \end{remark}

\section{Definition of \texorpdfstring{$\be$}{}-discriminants}\label{Sect:Definition}
In the following it will be natural to embed the domain $\Omega$ of the large-size limit of the dimer model into the complex projective space via
\begin{align*}
    \Omega \simeq \lbrace [1 : x_1 : x_2] \in \C P^2 \colon (x_1, x_2) \in \Omega \rbrace \subset \C P^2.
\end{align*}

For the remainder of the paper we will work in a more general setting by replacing $\C P^2$ with $\C P^n$, $n \geq 1$. Assume we are given a tuple $\vec{\eta} = (\eta_0, \dots, \eta_n)$ of meromorphic sections of some line bundle $\mathcal L \to \cR$ and consider the linear form
\begin{align*}
    \be(\vec{x}) = x_0 \eta_0 + \dots + x_n \eta_n = \langle \vec{x}, \vec{\eta} \rangle.
\end{align*}
The pole divisor $\mathsf{P}$ of $\be$ is then defined as follows:
\begin{definition}
    For a given tuple $\vec{\eta} = (\eta_0, \dots, \eta_n)$ the pole divisor $\sfP$ of $\be$ is given by
    \begin{align*}
        \mathsf{P} = \sum_{P \text{ is a pole of some } \eta_j} P.
    \end{align*}
\end{definition}
Note that even if an element $P \in \mathcal R$ is a pole of say $\eta_j, \eta_k$, it will still appear only once in the above sum. In case of higher order poles, the $P$ would appear with multiplicity $m$, where $m = \max_{j = 0, \dots, n}\lbrace \text{pole order of } \eta_j \text{ at } P\rbrace$.\footnote{By tensoring the line bundle $\mathcal L$ with the line bundle $[\mathsf{P}]$ one can essentially remove all poles and assume w.l.o.g that the $\eta_j$ are holomorphic sections, see Remark \ref{Rmk:Bundles}. However, we do not see any advantage in doing so.} 

Intuitively, $\mathsf{P}$ is just the pole divisor of $\be(\vec{x})$ for generic $\vec{x}$. We also denote by $\z$ the generic number of zeros of $\be$. Note that we have $\z = \p + \deg \mathcal L$ and in the special case of the canonical line bundle  $\z = \p + 2g-2$. The zero mapping $\mathsf{Z}$ associated to $\be$ is then defined via
\begin{align}\label{DefZ}
    \mathsf{Z} \colon \C P^n \to \Sym^{\z} \cR, \quad \mathbf x = [x_0 : \dots : x_n] \to \Div\big(x_0 \eta_0 + \dots + x_n \eta_n\big) + \sfP. 
\end{align}
Under Hypothesis \ref{HypoEta} below, $\mathsf{Z}$ is well-defined. Note that here we choose a representative $\vec{x} = (x_0, \dots, x_n) \in \C^{n+1}\setminus \lbrace 0 \rbrace$ for the element $\bx = [x_0 \colon \dots \colon x_n]$, but the choice does not matter. It is important to remember that elements of the divisor $\mathsf{Z}(\bx)$ are not necessary actual zeros of $\be(\vec{x})$, but could also coincide with the location of a pole in $\sfP$, whenever the order of the pole drops due to cancellation. We shall nonetheless refer to the points of the divisor $\mathsf{Z}(\bx)$ as the \emph{zeros of $\be(\vec{x})$}. Moreover, note that we do allow that all $\eta_j$ have a common zero. We call such zeros \emph{stationary zeros} of $\be(\vec{x})$, as they do not depend on $\bx \in \C P^n$. In fact, we only impose the following restrictions on the choice of $\be$:
\begin{hypothesis}\label{HypoEta}
    Assume that the meromorphic sections $\eta_0, \dots, \eta_n$ are linearly independent over $\C$ and do not have a common zero of higher multiplicity. 
\end{hypothesis}
The linear independence implies that $\mathsf{Z}$ is well-defined, as otherwise $\be(\vec{x})$ could vanish identically for certain values of $\bx \in \C P^n$. The property that the $\eta_j$ do not have common zeros of higher multiplicity is postulated, as otherwise the $\be$-discriminant would vanish identically.
\begin{remark}
We argue in Section \ref{Sect:Generealizations} that stationary zeros can be removed (or added) from the definition of $\mathsf{Z}$ without altering the construction of the underlying discriminant. Stationary zeros need to be removed from $\mathsf{Z}$ when computing arctic curves for quasi-periodic models of the hexagon with ramified covers, see Section \ref{SectHex}.
\end{remark}

\begin{remark} The image $\mathsf{Z}(\C P^n)$, in case no stationary zeros are present, is sometimes denoted by $|\be |$ and called \emph{linear system}, cf.~ \cite[Ch.~V.4]{Miranda}. This notion is used in the study of holomorphic mappings from Riemann surfaces into projective spaces, see e.g.~\cite[Prop.~4.15]{Miranda}. 
\end{remark}

The following definition generalizes the concept of an arctic curve.
\begin{definition}\label{DefAHS}
   For given meromorphic sections $\eta_1, \dots, \eta_n$ satisfying Hypothesis \ref{HypoEta}, we define the \emph{arctic hypersurface} as follows:
   \begin{align}\label{AHS}
       \mathcal A \mathcal H \mathcal S = \lbrace \bx \in \C P^n \colon \mathsf{Z}(\bx) \ \text{has a higher order zero} \rbrace.
   \end{align}
\end{definition}
Take now any metric $\dist \colon \cR \times \cR \to \R_{0 \leq }$ that is compatible with the complex structure on $\cR$, that is for a local chart $(\zeta, V)$ of $\cR$ we have that
\begin{align}\label{DefDist}
    \dist(R,S) \propto |\zeta(R)-\zeta(S)|, \qquad R,S \in V.
\end{align}
In other words, $\dist(R,S)/|\zeta(R)-\zeta(S)|$ defines a nonzero continuous function on $V \times V$. Note that this a local requirement in the sense that only the behavior of $\dist(R, S)$ for $R, S$ being close together is relevant. It also does not depend on the choice of the chart $(\zeta, V)$. One explicit example of such a metric in the genus $g \geq 1$ case is given by mapping $R, S \in \mathcal R$ via the Abel map to the Jacobian variety, and considering there the (complex) Euclidean distance (here one has to use that the Abel map is of maximal rank $1$ \cite[Lem.~7.8]{Schlag}).

Having chosen a metric $\dist( \cdot \, , \cdot)$, we can now state the main definition of the present paper (cf.~\cite[Ch.~9]{GKZ94}).

\begin{definition}[$\be$-discriminant]\label{DefDiscr}
    Let $\be = x_0\eta_0 + \dots + x_n\eta_n$, where the meromorphic sections $\eta_0, \dots, \eta_n$ satisfy Hypothesis \ref{HypoEta}, and let $\mathsf{Z}$ be given as in \eqref{DefZ}. Then we say that a homogeneous polynomial $H$ of degree $d$ in the variables $x_0, \dots, x_n$ is an \emph{$\be$-discriminant} if and only if \begin{align}\label{MainDef}
        |(x_0, \dots, x_n)|^{-d}|H(x_0, \dots, x_n)| \propto \prod_{i,j = 1, i \not = j}^{\# z} \dist(Z_i, Z_j),
    \end{align}
    for all $\bx = [x_0 \colon \dots \colon x_n] \in \C P^n$. In other words, the ratio of the two sides defines a continuous positive function on $\C P^n$. 
\end{definition}
Throughout this work we will denote $\be$-discriminants by $\Delta_{\be}$. It is clear that the (projective) zero set of $\Delta_{\be}$ coincides with the arctic hypersurface $\mathcal A \mathcal H \mathcal S$ from Definition \ref{DefAHS}.

Again we emphasize that the above definition does not depend on the choice of the metric $\dist( \cdot \, , \cdot)$. In fact, \eqref{MainDef} is in essence a local condition and can be alternatively stated using only local charts (though such a formulation would be more cumbersome). In the present paper we choose to work for simplicity with $\dist( \cdot \, , \cdot)$ to streamline certain arguments. 

Note that if it exists, an $\be$-discriminant is unique up to constant multiples. In fact, if $\Delta_{\be}$, $\widetilde \Delta_{\be}$ are both $\be$-discriminants (of possibly different degrees), then their ratio would be a meromorphic function on $\C^{n+1}$ with no zeros/poles except possibly at $x_0 = \dots = x_n = 0$. However, this is not possible for $n \geq 1$ as zeros/singularities cannot be isolated points of holomorphic functions in higher dimensions. In particular, it follows that the degree $d$ depends only on $\be$.

\medskip 

Before we continue with the \emph{construction}, let us briefly comment on the \emph{existence} of $\be$-discriminants. For this note that condition \eqref{MainDef} can be formulated in terms of pairwise disjoint local charts $(\zeta_k, V_k)$, $k = 1, \dots, M$ on $\mathcal R$, containing all the zeros $Z_j$, and substituting for $\dist(Z_i, Z_j)$ the expression $\zeta_k(Z_i) - \zeta_k(Z_j)$ provided $Z_i, Z_j$ are in the domain of the same chart $(\zeta_k, V_k)$ (if they are not, we can replace $\dist(Z_i, Z_j)$ simply by $1$). The resulting right-hand side would then be locally an analytic function of the variables $x_0, \dots, x_n$. This shows that locally the vanishing of an $\be$-discriminant (i.e.~coalescence of zeros) can be expressed as the vanishing of an analytic function. Hence Chow's Theorem \cite{Chow}, which states that closed \emph{analytic} varieties of $\C P^n$ are automatically \emph{algebraic}, implies that there exists a homogeneous polynomial in $x_0, \dots, x_n$ which vanishes if and only if two or more zeros of $\be(\vec{x})$ coalesce. One then has to only take care of possibly higher multiplicities of zeros to guarantee \eqref{MainDef}. As we will construct $\be$-discriminants in the subsequent section, we will not give more details regarding this type of existence argument. 

\section{Construction of \texorpdfstring{$\be$}{}-discriminants}\label{Sect:Construction}
The goal of this section is to generalize the construction of discriminants to the setting of arbitrary meromorphic sections of line bundles over $\cR$. We start with the familiar notion of a discriminant of a polynomial (the genus $0$ case) and continue from there. 
\subsection{Genus \texorpdfstring{$0$}{} case}\label{SectGenus0}
Let us begin with the case $\cR = \C P^1$. Assume we are given meromorphic sections $\eta_0, \dots \eta_n$ of some line bundle $\mathcal L \to \C P^1$ satisfying Hypothesis \ref{HypoEta}. Then in any chart of the form $(z, \C)$ excluding the point at infinity, the $\eta_j$ are represented by rational functions in $z$. Thus, $\be(\vec{x})$ can be written as $\big[\sum_{j =0}^n x_j q_j(z)\big]/r(z)$ for some polynomials $r(z), q_0(z), \dots, q_n(z)$ which do not have a common factor. Without loss of generality we can assume that our chart is chosen such that $\infty$ is not a common zero of all $\eta_j$. In other words, $\infty$ is not a stationary zero of $\be(\vec{x})$. Then it follows that the pole divisor $\sfP$ is given by the zeros of $r$ in $\C$ together with a nonnegative multiple of infinity. More precisely, we have
\begin{align*}
    \sfP = \text{Zeros}_\C(r) + [\deg(\vec{q}) - \deg(r) - \deg(\mathcal L)] \cdot \infty,
\end{align*}
where $\deg(\vec{q}) = \max_{j = 0, \dots, n} \lbrace \deg(q_j) \rbrace$. Note that $\deg(\vec{q}) - \deg(r) - \deg(\mathcal L) \geq 0$, as otherwise $\be(\vec{x})$ would have a stationary zero at infinity.

The divisor of $\be$ is given by
\begin{align*}
    \text{Divisor}\big(x_0 \eta_0 + \dots + x_n \eta_n\big) = \text{Zeros}_\C\big(x_0 q_0 + \dots + x_n q_n\big) - \text{Zeros}_\C(r) - m\cdot \infty, 
\end{align*}
where $m = m(\vec{x}) = \deg\big(x_0 q_0 + \dots x_n q_n \big) - \deg(r) - \deg(\mathcal L)$. Thus we obtain
\begin{align*}
    \mathsf{Z}(\bx) = \text{Zeros}_\C\big(x_0 q_0 + \dots + x_n q_n\big) + \big[\deg(\vec{q}) - \deg\big(x_0 q_0 + \dots + x_n q_n \big)\big]\cdot \infty, 
\end{align*}
In particular, $\# z = \deg(\vec{q})$. We see that (possibly multiple) zeros escape to infinity exactly when the degree of $x_0 q_0 + \dots + x_n q_n$ becomes lower than the maximal degree $\deg(\vec{q})$ of the $q_j$. It is also clear that the choice of holomorphic line bundle $\mathcal L$ plays no significant role for the zero map $\mathsf{Z}(\bx)$, and this remains true in the higher genus case as well, cf.~Remark \ref{Rmk:Bundles}.  

To define the $\be$-discriminant, let us denote by $\Delta_{N} \colon \C^{N+1} \to \C$ the usual discriminant of polynomials of degree $N$. Then $\Delta_{N}$ is a homogeneous polynomial of degree $2N-2$ and for a polynomial $\mathcal P(z) = \sum_{k=0}^{N} c_k z^k$ with $c_{N} \not = 0$ it satisfies
\begin{align*}
    \Delta_{N}(\mathcal P) = \Delta_{N}(c_0, \dots, c_{N}) = (-1)^{\frac{N(N+1)}{2}}c_{N}^{2N-2} \prod_{i, j = 1, \, i \not = j}^{N} (z_i-z_j), 
\end{align*}
where $z_1, \dots, z_{N}$ are the zeros of $\mathcal P(z)$. 

Let us rewrite $\sum_{j = 0}^n x_j q_j(z) = \sum_{k = 0}^{\z} c_k(\vec{x}) z^k = \mathcal Q(z; \vec{x})$ with $c_k(\vec{x}) = \sum_{j = 0}^n x_j q_{j, k}$, where $q_{j,k}$ is the $k$-th coefficient of $q_j(z)$ (which is zero if $\deg q_j < k$). We claim the following.
\begin{proposition}
    Let $\be = x_0\eta_0 + \dots + x_n\eta_n$, with $\eta_0, \dots, \eta_n$ satisfying Hypothesis \ref{HypoEta}. If $\mathcal R$ is of genus $0$, then the $\be$-discriminant is given by the homogeneous polynomial $\Delta_{\z}(c_0(\vec{x}), \dots, c_{\z}(\vec{x}))$:
    \begin{align}\label{Deltaz}
        \Delta_{\be}(\vec{x}) = \Delta_{\z}(c_0(\vec{x}), \dots, c_{\z}(\vec{x})).
    \end{align}
\end{proposition}
In particular, the $\be$-discriminant $\Delta_{\be}$ is a homogeneous polynomial of degree $2 \cdot \z -2$ where $\z = \deg(\sfP) + \deg(\mathcal L)$.
\begin{proof}
Let us first define the metric $\dist(\cdot \, , \, \cdot)$. One can use for example the identification $\C P^1 \sim S^2$, and define the metric via the geodesic distance on $S^2$. The exact choice is not relevant as long as \eqref{DefDist} is satisfied. 

Consider a fixed $\bx \in \C P^n$. We will show that $\Delta_{\be}$ defined in \eqref{Deltaz} satisfies the discriminant property \eqref{MainDef} in a neighbourhood of $\bx$. First note that the polynomial discriminant satisfies $\Delta_N(\mathcal P(z - w)) = \Delta_N(\mathcal P(z))$ and, provided $\mathcal P(0) \not = 0$, $\Delta_{N}(\mathcal P^*(z)) = \Delta_N(\mathcal P(z))$ where $\mathcal P^*(z) = z^N \mathcal P(z^{-1})$ is the reciprocal polynomial. Thus in case $\infty$ is a zero of $\be(\vec{x})$ (that is $\mathcal Q(z; \vec{x})$ satisfies $\deg(\mathcal Q(\, \cdot \,; \vec{x})) < \z = \deg(\vec{q})$), we can first translate by $w \not = z_j$ in case $\mathcal Q(0; \vec{x}) = 0$ and then consider instead the reciprocal polynomial:
\begin{align*}
    \Delta_{\z}(\mathcal Q(z; \vec{x})) = \Delta_{\z}(\mathcal Q^*_w(z; \vec{x})) = (-1)^{\frac{\z(\z+1)}{2}} \widetilde{c}_{\z}^{\, 2\cdot \z-2} \prod_{i,j = 1, i \not = j}^{\z}(\widetilde z_i- \widetilde z_j) 
\end{align*}
for some $\widetilde{c}_{\z}^{\, 2\cdot \z-2} = \widetilde{c}_{\z}^{\, 2\cdot \z-2}(\vec{x}) \not = 0$ and $\mathcal Q_w(z; \vec{x}) = \mathcal Q(z-w; \vec{x})$. Here, $\widetilde z_j$ with $j = 1, \dots, \#z$ denote the zeros of $\mathcal Q_w^*(z; \vec{x})$ which is a polynomial of degree $\z = \deg(\vec{q})$. Now in a sufficiently small neighbourhood of $\vec{x}$ (and hence of $\bx$), the leading coefficient $\widetilde{c}_{\z}^{\, 2\cdot \z-2}$ will not vanish, meaning no zero will escape to infinity. However, staying away from infinity we know that $\dist(Z_i, Z_j) \propto |\widetilde z_i - \widetilde z_j|$. As additionally $|\vec{x}| \propto |\widetilde{c}_{\z}|$, it follows that $|\vec{x}|^{-2\cdot \z+2} \Delta_{\z}(\mathcal Q(z; \vec{x}))$ satisfies \eqref{MainDef} finishing the proof. 
\end{proof}
\subsection{Higher genus case}
In the following we will briefly summarize the main steps in the construction of $\be$-discriminants in the higher genus case. 

The first step is to substitute in place of $(z_i-z_j)$ with $z_i, z_j \in \C$ in \eqref{Def:Delta_n} the prime function $E_{\ve}(Z_i, Z_j)$ with $Z_i, Z_j \in \mathcal R$ being zeros of $\be(\vec{x})$. This leads us to the definition of the symmetric product $\Pi_1$ in \eqref{Pi1} where we additionally apply the Abel map to the zeros of $\be(\vec{x})$.

Immediately two problems arise with this approach. On the one hand the product $\Pi_1$ is multi-valued, and on the other it has additional spurious zeros. Both issues are related to the use of Riemann theta functions, though the spurious zeros only appear in the genus $g \geq 2$ case.

To arrive at a single-valued expression we introduce another symmetric product $\Pi_2$ in \eqref{Pi2} and study the ratio $\ff = \Pi_1/\Pi_2$.  The function $\ff$ vanishes when two zeros coalesce and turns out to be a meromorphic function on $\C P^n$ by Proposition \ref{Prop}. However, it still has the spurious zeros from before together with additional poles coming from the zeros of $\Pi_2$. 

Fortunately, these additional zeros and poles can be easily taken care of via the linear functions $L^Q$, $Q \in \mathcal R$ defined in \eqref{defL}. This essentially follows from Lemma \ref{LemG}. This lemma allows us to cancel the unwanted zeros and poles and prove our main result on $\be$-discriminants in Theorem \ref{mainTheorem}. In particular, we obtain an explicit formula for a homogeneous polynomial $\Delta_{\be}$ which satisfies \eqref{MainDef}. Through this formula we can easily read off the degree of $\Delta_{\be}$ as a function of the genus $g$ of $\mathcal R$ and the number of poles (or zeros) of $\be(\vec{x})$.

\medskip 

Let us now describe each of the above steps in detail. First, choose for $\cR$ a canonical basis of cycles $\lbrace \mathbf{a}_j, \mathbf{b}_j\rbrace_{j=1}^g$ and a corresponding basis of holomorphic differentials $\lbrace \omega_i \rbrace_{i=1}^g$ such that
\begin{align*}
    \int_{\mathbf{a}_j} \omega_i = \delta_{ij}.
\end{align*}
In applications to dimer models one would choose for the $\mathbf{a}$-cycles the $g$ real ovals $X_1, \dots, X_g$ not containing the poles of $\be$. Let us also define the \emph{period matrix} $B = (B_{i,j})_{i,j=1}^g$ to be 
\begin{align*}
    B_{ij} = \int_{\mathbf{b}_j} \omega_i.
\end{align*}
Abbreviating $\vec{\omega} = (\omega_1, \dots, \omega_g)^T$ we introduce the \emph{multi-valued Abel map} $u$ with base point $P_0 \in \cR$ via
\begin{align*}
    u \colon \cR \to \C^g, \quad Q \mapsto \int_{P_0}^Q \vec{\omega}. 
\end{align*}
Note that we make the rather unusual choice of working with a multi-valued Abel map instead of a single-valued one with range $\C^g / \Lambda$, where $\Lambda = \Z^g + B \Z^g$. The reason for this will become apparent in a moment.

Next we introduce the Riemann theta function (see \cite[Ch.~II]{Tata1})
\begin{align*}
    \theta(\vs \, | B) = \sum_{\vm \in \Z^g} \exp\Big\lbrace 2\pi\I\Big( \frac{1}{2}\langle \vm, B\vm\rangle +\langle \vm, \vs\rangle \Big)\Big\rbrace, \qquad \vs \in \C^g, \, B \in \C^{g \times g},
\end{align*}
which is a higher genus generalization of the Jacobi theta function. It satisfies the following quasi-periodicity property which is central to the present work (see \cite[\href{https://dlmf.nist.gov/21.3}{Eq.~(21.3.3)}]{DLMF}):
\begin{align}\label{QuasiPer}
   \theta(\vs+\vm_1 + B \vm_2 \, | B) &= \exp\Big\lbrace -2\pi\I\Big( \frac{1}{2}\langle \vm_2, B \vm_2\rangle +\langle \vm_2, \vs\rangle \Big)\Big\rbrace \theta(\vs \, | B)
   \\\nonumber
   & \hspace{5.6cm} \vm_1, \vm_2 \in \Z^g.
\end{align}
We will drop the dependence on the period matrix $B$ which can be viewed as fixed. 

 Now consider functions $\Ee$, for some $\ve \in \C^g$, of the form 
\begin{align*}
    \Ee(R, S) = \theta\Big(\int_R^S \vo + \ve\Big), \qquad R, \, S \in \cR.
\end{align*}
For details see  \cite[Ch.~II.\S 3]{Tata1}. We will call these \emph{prime functions}, to distinguish them from the closely related prime forms, see e.g.~\cite[Ch.~II]{Fay}, \cite[p.~3.207]{Tata2}. Note that $\Ee$ defines a multi-valued function on $\cR \times \cR$ due to the quasi-periodicity property \eqref{QuasiPer}, meaning that the choice of integration contour from $R$ to $S$ matters. 

Closely related to the prime function we also introduce the following \emph{single-valued} function on $\C^g \times \C^g$:
\begin{align}\label{PrimeFunction2}
    \vte(\vr, \vs) = \theta\big(\vs - \vr + \ve\big), \qquad \vr, \vs \in \C^g.
\end{align}
Note however that $\vte$ remains multi-valued on $\C^g / \Lambda$. 

Next, we choose $\ve \in \C^g$ such that $\Ee \not \equiv 0$ and $\theta(\ve) = \Ee(R,R) = 0$. That such an $\ve$ exists, and is in fact generic among theta divisors, is shown in \cite[Ch.~II.~Lem.~3.3]{Tata1}. It turns out that for any choice of $\ve$ with $\Ee \not \equiv 0$, there exists points $R_1, \dots, R_{g-1}$ and $S_1, \dots, S_{g-1}$ such that $\Ee(R, S)$ for $S \not = S_j$ has, as a function of $R$, exactly $g$ zeros, namely at $S$ and $R_1, \dots R_{g-1}$ (or higher order zeros if $R_i = R_j$ ect.). Similarly, if we fix $R \not = R_j$ then $\Ee(R,S)$, viewed as a function of $S$, has exactly $g$  zeros at $R$ and $S_1, \dots, S_{g-1}$ (or higher order zeros if $S_i = S_j$ ect.), see \cite[Ch.~II.~Lem.~3.4]{Tata1}. We thus see that $\Ee(Z_i,Z_j)$ has some similarities with the expression $(z_i - z_j)$ in \eqref{Def:Delta_n}, as generically it has simple zeros at the diagonal. However, two crucial differences remain:
\begin{itemize}
    \item $\Ee$ is multi-valued;
    \item $\Ee$ has additional \emph{spurious} zeros if the first argument is in the (multi)set $\mathbf{R} = \lbrace R_1, \dots, R_{g-1} \rbrace$, or the second argument is in the (multi)set $\mathbf{S} = \lbrace S_1, \dots, S_{g-1} \rbrace$. 
\end{itemize}
In the following construction we will address both these difficulties. We start with the issue of multi-valuedness. 

For any $\mathsf{Z}(\bx)$, let us choose representatives $\vz_j = u(Z_j)\in \C^g$, $j = 1, \dots, \z$, with the only requirement that 
\begin{align}\label{SigmaRule}
    \vz_1 + \dots + \vz_{\z} = \Sigma,
\end{align}
where $\Sigma \in \C^g$ is a fixed representative satisfying $\Sigma \equiv u(\mathsf{Z}(\bx)) \ \text{mod} \ \Lambda$ (which is independent of $\bx \in \C P^n$ by Abel's Theorem). 
At the moment there is no further requirement that the choice of representatives/the order of the elements $\vz_j$ is in any sense continuous with respect to $\bx \in \C P^n$.

Now consider the function $\Pi_1 \colon \C^{g \times \z} \to \C$ given by
\begin{align}\label{Pi1}
    \Pi_1(\vz_1, \dots, \vz_{\z}) = \prod_{i\not = j} \vte(\vz_i, \vz_j) \in \C.
\end{align}
Observe that as $\Pi_1$ is symmetric with respect to its arguments it is well-defined as a function of the (multi)set $\lbrace \vz_1, \dots, \vz_{\z} \rbrace \subset \C^g$. However, it is not a function of the divisor $\mathsf{Z}(\bx)$, as we had to choose the representatives $\vz_j \in \C^g$ beforehand. This remains true even if we take into account \eqref{SigmaRule}. In fact, consider the transformation
\begin{align*}
     ( \vz_1, \dots, \vz_{\z} ) \mapsto ( \vz_1, \dots, \vz_i + \vec{B}_m, \dots, \vz_j - \vec{B}_m, \dots, \vz_{\z} ),
\end{align*}
where $\vec{B}_m =(B_{1,m}, \dots, B_{g,m})^T$, $m \in \lbrace 1, \dots, g\rbrace$ denote the columns of the period matrix $B$. This transformation clearly respects condition \eqref{SigmaRule}, however from \eqref{QuasiPer}, \eqref{PrimeFunction2} we have
\begin{align*}
    \Pi_1(\vz_1, &\dots, \vz_i - \vec{B}_m, \dots, \vz_j + \vec{B}_m, \dots, \vz_{\z}) 
    \\
    &=\exp \Big\lbrace - 4\pi \I\big(B_{m,m} + (\vz_i-\vz_j)_m\big)\# z\Big\rbrace \ \Pi_1(\vz_1, \dots, \vz_{\z}),
\end{align*}
where $B_{j,k}$ denotes the entries of $B$.

To correct for this discrepancy choose an arbitrary element $Q \in \mathcal R \setminus \mathbf{S}$ together with a representative $\vq = u(Q) \in \C^g$.  Define the function $\Pi_2 \colon \C^{g \times \z} \to \C$ via
\begin{align}\label{Pi2}
    \Pi_2(\vz_1, \dots, \vz_{\z} \, ; \, \vq \, ) = \prod_{i=1}^{\z} \vte(\vz_i, \vq \, ) \in \C.
\end{align}
We now find that
\begin{align}\label{Pi2Mono}
    \Pi_2(\vz_1, &\dots, \vz_i - \vec{B}_m, \dots, \vz_j + \vec{B}_m, \dots, \vz_{\z} \, ; \, \vq \, ) 
    \\\nonumber
    &=\exp \Big\lbrace - 2\pi \I\big(B_{m,m} + (\vz_i-\vz_j)_m\big)\Big\rbrace \ \Pi_2(\vz_1, \dots, \vz_{\z}\, ; \, \vq \, ).
\end{align}
Hence, we see that as long as the denominator does not vanish, the function defined via 
\begin{align*}
    \ff(\bx \, ; \, \Sigma, \, \vec{e}, \, \vec{q} \,) = \frac{\Pi_1(\vz_1, \dots, \vz_{\z})}{[\Pi_2(\vz_1, \dots, \vz_{\z}\, ; \, \vq \, )]^{2\, \cdot\,\z}}, \qquad \text{with} \ \vz_1 + \dots + \vz_{\z} = \Sigma
\end{align*}
is well-defined. Moreover, one can check that the relation  
\begin{align}\label{SigmaSigma'}
    \ff(\bx \, ; \Sigma', \, \vec{e}, \, \vec{q} \,) = \exp\Big\lbrace 2\pi\I\Big( B_{m,m} + 2 \Sigma_m - 2(\ve + \vq \, )_m \cdot \# z\Big)\Big\rbrace \, \ff(\bx \, ; \Sigma, \, \vec{e}, \, \vec{q}\,)
\end{align}
holds whenever $\Sigma' = \Sigma + \vec{B}_m$. These observations allow us to prove the following proposition.
\begin{proposition}\label{Prop}
    Assume that the meromorphic sections  $\eta_0, \dots, \eta_n$ satisfy Hypothesis \ref{HypoEta}. Let $\ve \in \C^g$ be chosen such that for any common zero $C_0 \in \cR$ of $\eta_j$, $j = 0, \dots, n$, we have $C_0 \not \in \mathbf{R} \cup \mathbf{S}$. Now choose any $Q \in \cR \setminus \mathbf{S}$ which is not a common zero of the $\eta_j$, and let $\vec{q} = u(Q)$. Consider the mapping
    \begin{align*}
        \ff( \ \cdot \ ; \Sigma, \, \vec{e}, \, \vec{q} \,) \colon \C P^n \to  \C \cup \lbrace \infty \rbrace, \quad \bx \mapsto \frac{\Pi_1(\vz_1, \dots, \vz_{\z})}{[\Pi_2(\vz_1, \dots, \vz_{\z}; \, \vec{q}\,) ]^{2\, \cdot\, \z}},
    \end{align*}
    where the representatives $\vec{z}_j = u(Z_j)$ with $\mathsf{Z}(\bx) = \sum_{j=1}^{\z} Z_j$ are chosen such that $\vz_1 + \dots + \vz_{\z} = \Sigma \in \C^g$. Then $\ff( \ \cdot \ ; \Sigma, \, \vec{e}, \, \vec{q} \,)$ defines a meromorphic function on $\C P^n$.
\end{proposition}
\begin{proof}
First, let us remark that $\vec{e}$ can always be chosen such that $C_0 \not \in \mathbf{R} \cup \mathbf{S}$ for all stationary zeros $C_0$ by \cite[Ch.~II.~Lem.~3.3]{Tata1}.
    
    Recall that a function on a complex manifold is meromorphic, if and only if it can be locally written as a ratio $h_1/h_2$ of two holomorphic maps with $h_2 \not \equiv 0$ (see \cite[Def.~2.1.8]{Huybrechts}). Take an arbitrary point $\bx \in \C P^n$ and let us write $\mathsf{Z}(\bx) = \sum_{k=1}^M r_k Z_k$, where $M$ is the number of distinct zeros, i.e. $Z_j \not = Z_k$ for $j \not = k$. We then have $\z = \sum_{k=1}^M r_k$. Now choose arbitrary representatives $\vz_k = u(Z_k) \in \C^g$. Note that here we do not assume \eqref{SigmaRule}, i.e. we only have 
    \begin{align*}
        \sum_{k=1}^M r_k \vz_k = \Sigma',
    \end{align*}
    with $\Sigma' \equiv \Sigma \mod \Lambda$. In fact, if $r_k > 1$ for $k = 1, \dots, M$, we might not have the freedom to choose $\vz_1, \dots, \vz_M \in \C^g$ such that \eqref{SigmaRule} is satisfied. However, due to condition \eqref{SigmaSigma'} it is enough to show that $\ff(\, \cdot \, ; \Sigma' , \vec{e}, \, \vec{q} \,)$ is locally the ratio of two holomorphic functions, so we can w.l.o.g.~assume that $\sum_{k=1}^M r_k \vz_k = \Sigma$.  
    
    Note that as by definition all $Z_k$ are pairwise distinct, we can find $M$ small disjoint neighbourhoods $U_k \subset \C^g$, which are also disjoint as subsets of $\C^g / \Lambda$, such that $\vz_k \in U_k$ and $\vec{v}_1, \vec{v}_2 \in U_k$ with $\vec{v}_1 \equiv \vec{v}_2 \mod \Lambda$ already implies $\vec{v}_1 = \vec{v}_2$ in $\C^g$.
    
    Now, assume w.l.o.g.~that $x_0 \not = 0$ and choose a neighbourhood $N$ of $\bx = [x_0 : \dots : x_n]$ of the form $N = \lbrace [x_0 : x_1 + t_1 : \dots : x_n + t_n] : \vec{t} = (t_1, \dots, t_n) \in T \rbrace$, where $T \subset \C^n$ is a small open neighbourhood of $\vec{0}$ that will be determined later. Note that $\vec{t} \mapsto [x_0 : x_1 + t_1 : \dots :  x_n + t_n]$ defines a local chart on $\C P^n$. 
    
    Provided $T$ is small enough, we obtain $M$ mappings
    \begin{align*}
        \psi_k : T \to \Sym^{r_k} U_k, \quad \vec{t} \mapsto \vz_k^{\, (1)} + \dots + \vz_k^{\, (r_k)},
    \end{align*}
    such that $u(\mathsf{Z}(\bx)) =  \sum_{k=1}^M (\vz_k^{\, (1)} + \dots + \vz_k^{\, (r_k)})$ and the vectors $\vz_1^{\, (1)}, \dots, \vz_M^{\, (r_M)}$ are the images of the perturbed zeros of $\be$ under the Abel map. In other words, under sufficiently small perturbations the zero $Z_k \in \cR$ of multiplicity $r_k$ can split into $r_k$  points (not necessarily pairwise distinct), which, under the Abel map, have unique representatives in the sets $U_k \subset \C^g$. In particular, we have
    \begin{align*}
        \sum_{k=1}^M (\vz_k^{\,(1)} + \dots + \vz_k^{\, (r_k)}) =  \Sigma
    \end{align*}
    due to continuity of the maps $\psi_k$.
    
    Importantly, with these restrictions the choice of the $\z$ vectors $\vz_1^{\, (1)}, \dots, \vz_M^{\, (r_M)} \in \C^g$ (counting multiplicity) is unique up to permutations. So we know that the mappings
    \begin{align*}
        h_1 \colon T &\to \C,  \quad \vec{t} \mapsto \Pi_1(\vec{z}_1^{\, (1)}, \dots, \vec{z}_M^{\, (r_M)}) \in \C, 
        \\
        h_2 \colon T &\to \C,  \quad \vec{t} \mapsto \Pi_2(\vec{z}_1^{\, (1)}, \dots, \vec{z}_M^{\, (r_M)}) \in \C,
    \end{align*}
    are well-defined continuous functions. We remark that $h_2 \not \equiv 0$ due to our choice of $\ve$ and $\vec{q}$. 
    
    Now if for some $\vec{t}_0 \in T$ all $\vz_1^{\, (1)}, \dots, \vz_M^{\, (r_M)}$ are distinct, it immediately follows that $h_1, h_2$ are locally holomorphic functions, as each $\vz_k^{\, (j)} = \vz_k^{\, (j)}(t_1, \dots, t_n)$ can locally near $\vec{t}_0$ be written as a holomorphic function of $t_1, \dots, t_n$ by the implicit function theorem. The general case follows from a simple local analysis using charts (see Lem.~\ref{LemG} for details). This implies for $t_2, \dots, t_n$ fixed that the functions $h_1$ and $h_2$ are holomorphic in the variable $t_1$ (though the individual zeros are in general not). The same conclusion holds for any other $t_j$. Thus, Hartog's theorem on separate holomorphicity implies that $h_1, h_2$ are holomorphic, finishing the proof. 
\end{proof}
From the properties of $\Ee$ it follows that
\begin{align}\label{Pipropto}
    \Pi_1 \propto \prod_{i \not = j} \dist(Z_i, Z_j) \prod_{i, k} \dist(Z_i, R_k)^{\z-1}\prod_{j, k} \dist(Z_j, S_k)^{\z-1}
\end{align}
and
    \begin{align}\label{Pi2Propto}
        \Pi_2 \propto \prod_{i, k} \dist(Z_i, R_k) \prod_{i} \dist(Z_i, Q), 
    \end{align}
    with $i, j \in \lbrace 1, \dots, \#z \rbrace$ and $k \in \lbrace 1, \dots, g-1 \rbrace$.
    In particular, we know exactly the zero/pole structure of the meromorphic function $\ff$. We will use this knowledge to extract from $\ff$ the $\be$-discriminant $\Delta_{\be}$ in Theorem \ref{mainTheorem}.

    \medskip
    
Let us now turn to the issue of spurious zeros (and poles) of $\ff$. The following lemma will play a key role in this undertaking.
\begin{lemma}\label{LemG}
Let $D \subset \C$, $U \subset \C^{n+1} \setminus \lbrace 0 \rbrace$ be  open sets and consider the function $L \colon D \times U \to \C$ given by
    \begin{align*}
        L(\zeta; x_0, \dots , x_n) = x_0 k_0(\zeta) + \dots +  x_n k_n(\zeta) = \langle \vec{x}, \vec{k}(\zeta) \rangle, 
    \end{align*}
 where the $k_j \colon D \to \C$ are holomorphic and $k_j(\widehat \zeta) \not = 0$ for some fixed $\widehat \zeta \in D$ and at least one $j \in \lbrace 0, \dots, n \rbrace$. Let us further assume that for any choice of parameters $\vec{x} = (x_0, \dots, x_n) \in U$, the function $L( \ \cdot \ ;\vec{x}) \colon D \to \C$ has exactly $d$ zeros $\zeta_i = \zeta_i(\vec{x}) \in D$ with $i =1, \dots, d$ (the ordering does not matter). Then the mapping
 \begin{align}\label{DefG}
     G_{\widehat \zeta} \colon U \to \C, \quad G_{\widehat \zeta}(\vec{x}) = \frac{\prod_{i = 1}^d(\zeta_i(\vec{x}) - \widehat \zeta)}{L(\widehat \zeta ; \vec{x})}
 \end{align}
 defines a nonvanishing holomorphic function on $U$. 
\end{lemma}
\begin{proof} 

 We start by showing that $G_{\widehat \zeta}$ is a meromorphic function, i.e.~the numerator and denominator on the right-hand side of \eqref{DefG} are holomorphic. Clearly, $L(\widehat \zeta; \vec{x})$ is a holomorphic function of $\vec{x}$ as it is linear. Moreover, $L(\widehat \zeta; \vec{x}) \not \equiv 0$ as $k_j(\widehat \zeta) \not = 0$ for some $j$ by assumption.
 
 Now let us consider the numerator. Due to Hartog's theorem on separate holomorphicity it is enough to show holomorphicity with respect to each variable separately, hence let us focus on the variable $x_0$. As long as $\zeta_i \not = \zeta_j$ for $i \not = j$, the individual $\zeta_i = \zeta_i(x_0)$ are locally holomorphic functions of $x_0$ by the implicit function theorem. Hence, it follows that in this case product $\prod_{i = 1}^d(\zeta_i(x_0) - \widehat \zeta)$ is locally holomorphic in $x_0$. 
 
 In the case of a zero $\zeta_0$ of multiplicity $m$ for some critical $x_{0,c}$, let us rewrite $L(\zeta; \vec{x})$ as a linear function of $x_0$ with $x_1, \dots, x_n$ fixed:
    \begin{align}\label{tk}
       L(\zeta; \vec{x}) =  x_0 k_0(\zeta) + K(\zeta) = (x_0-x_{0,c}) \underbrace{(\zeta-\zeta_0)^p\ell(\zeta)}_{= \ k_0(\zeta)} + \underbrace{(\zeta-\zeta_0)^m H(\zeta)}_{= \ x_{0,c}k_0(\zeta)+K(\zeta)}
    \end{align}
 with $\ell, H$ holomorphic, $\ell(\zeta_0), H(\zeta_0) \not = 0$, and $K(\zeta) = x_1 k_1(\zeta) + \dots + x_n k_n(\zeta)$. We see that $\min \lbrace p, m \rbrace$ zeros at $\zeta_0$ are independent of $x_0$, while in case $m > p$ the $m-p$ remaining zeros $\widetilde \zeta_b$ in the vicinity of $\zeta_0$ can be approximated via $\widetilde \zeta_b \simeq \zeta_0 + (x_0- x_{0,c})_b^{\frac{1}{m-p}} (-\ell(\zeta_0)/H(\zeta_0))^{\frac{1}{m-p}}$, where the extra subscript $b = 1, \dots, m-p$ indicates the different branches of the $(m-p)$th root. Again, the product $\prod_{i = 1}^d(\zeta_i(t) - \zeta)$ remains holomorphic in $x_0$ (use e.g.~Riemann's theorem on removable singularities) with $\prod_{b = 1}^{m-p} (\widetilde \zeta_b(x_0) - \zeta_0) \propto (x_0 - x_{0,c})$ as $x_0 \to x_{0,c}$.

 We thus see that $G_{\widehat \zeta}$ is meromorphic. It remains to show that it is a holomorphic nonvanishing function. By assumption the vector $\vec{k}(\widehat \zeta) = (k_0(\widehat \zeta), \dots, k_n(\widehat \zeta))^T$ is nonzero, hence there exists a real invertible matrix $C$ such that $C\vec{k}(\widehat \zeta)$ is a vector with \emph{all} entries nonzero. As $L(\zeta; \vec{x}) = \langle (C^{-1})^{T}\vec{x}, C\vec{k}(\zeta)\rangle$, and $\vec{x} \mapsto \vec{y} = (C^{-1})^{T}\vec{x}$ is a linear, hence \emph{holomorphic}, change of variables, showing holomorphicity of $G_{\widehat \zeta}$ with respect to $\vec{x}$ and $\vec{y}$ is equivalent. It follows that we can w.l.o.g.~assume that $k_j(\widehat \zeta) \not = 0$ for all $j = 0, \dots, n$. This will simplify the analysis that follows.

 Note that by construction the denominator of $G_{\widehat \zeta}$ vanishes if and only if  $L(\widehat \zeta ; \vec{t})$ vanishes. Thus we have to show that the order of the vanishing is equal. By the previous arguments we can assume $k_0(\widehat \zeta) \not = 0$, which implies that $L(\widehat \zeta ; \vec{x})$ viewed as a function of $x_0$ has at most one simple zero. We assume that it exists and denote it by $\widehat x_0$ (as otherwise there is nothing to show). We write similarly as before
 \begin{align*}
     L(\zeta; \vec{x}) =  x_0 k_0(\zeta) + K(\zeta) = (x_0-\widehat x_0) k_0(\zeta) + \underbrace{(\zeta-\widehat \zeta)^{m} H(\zeta)}_{= \ \widehat x_0k_1(\zeta)+K(\zeta)}.
 \end{align*}
Here, our assumption $k_0(\widehat \zeta) \not = 0$ implies that $p = 0$ in \ref{tk}, i.e.~$k_0(\zeta) = \ell(\zeta)$. If we denote by $\widehat \zeta_j$ the $m$ zeros close to $\widehat \zeta$, then $\prod_{j = 1}^{m} (\widehat \zeta_j(x_0) - \widehat \zeta) \propto (x_0 - \widehat x_0)$, implying that we exactly cancel the simple zero at $\widehat x_0$ of $L(\widehat \zeta; \vec{x})$ viewed as a function of $x_0$. This finishes the proof.   
\end{proof}
The previous Lemma \ref{LemG} motivates the following definition. For any point $P \in \cR$ which is not a stationary zero of $\be$, let us choose a local chart $(\zeta, V)$ with $P \in V$. Assume that in this chart $x_0 \eta_0 + \dots + x_n \eta_n$ can be expressed as $x_0k_0(\zeta) + \dots  + x_n k_n(\zeta)$ with some holomorphic functions $k_j$ on $\zeta(V)$ and $k_j(\zeta_P) \not = 0$ for at least one $j$. We then define the linear form
\begin{align}\label{defL}
    L^P(x_0, \dots, x_n) = x_0 k_0(\zeta_P) + \dots  + x_n k_n(\zeta_P), \qquad \zeta_P = \zeta(P).
\end{align}
\begin{remark}\label{RemarkL}
In case $P$ is a pole of $\be$ of order $m$, we have to adjust the definition of $L^P$ by replacing the meromorphic $k_j(\zeta)$ with $(\zeta-\zeta_P)^m k_j(\zeta)$. More generally, we are only concerned with $[k_0(\zeta_P) \colon \dots \colon k_n(\zeta_P)] \in \C P^n$, which always exists by analytic continuation even if $P$ is a pole of $\be$. Similarly, different charts would result in the same $L^P$ up to some constant nonzero multiple. As we will be only interested in the zero set of $L^P$, we will not make this chart dependence explicit.
\end{remark}
Recall that $\mathbf{R} = \lbrace R_1, \dots, R_{g-1} \rbrace$, $\mathbf{S}= \lbrace S_1, \dots, S_{g-1} \rbrace$ denote the sets of points in $\mathcal R$ for which $\Ee(R_j, \, \cdot \,) \equiv 0$ and $\Ee(\, \cdot \,, S_j) \equiv 0$. Moreover, in Proposition \ref{Prop} we have assumed that the common zeros of $\eta_j$ are disjoint from $\mathbf{R} \cup \mathbf{S} \cup \lbrace Q \rbrace$. With this in mind we are now ready to prove our main result.
\begin{theorem}\label{mainTheorem}
Assume the meromorphic sections $\eta_0, \dots \eta_n$ satisfy Hypothesis \ref{HypoEta} and let $\delta_{\be}$ be defined as in Proposition \ref{Prop}. Then the function $\Delta_{\be} \colon \C^{n+1} \setminus \lbrace 0 \rbrace \to \C$ given by 
\begin{align}\nonumber
    \Delta_{\boldsymbol{\eta}}(x_0, \dots, x_n) &= \ff([x_0 : \dots : x_n]; \Sigma, \, \vec{e}, \, \vec{q} \,) 
    \\\label{DiscrDef}
    &\times \frac{\Big(\prod_{j = 1}^{g-1} L^{R_j}(x_0, \dots, x_n)\Big)^{\# z + 1} \Big( L^Q(x_0, \dots, x_n) \Big)^{2 \, \cdot \, \# z}}{\Big(\prod_{j = 1}^{g-1} L^{S_j}(x_0, \dots, x_n)\Big)^{\# z-1}}
\end{align}
defines a homogeneous polynomial of degree $2g - 2 + 2\cdot\# z$. The polynomial $\Delta_{\boldsymbol{\eta}}$ is the, up to multiples unique, $\be$-discriminant in the sense of Definition \ref{DefDiscr}, that is, it satisfies \eqref{MainDef}. In particular, the zero set of $\Delta_{\boldsymbol{\eta}}$ viewed as a subset of $\C P^n$ coincides with the arctic hypersurface $\mathcal A \mathcal  H \mathcal S$ generated by  $\eta_0, \dots, \eta_n$ (see Def.~\ref{DefAHS}). 
\end{theorem}
\begin{proof}
    Fix a vector $(\widetilde x_0, \dots, \widetilde x_n) \in \C^{n+1} \setminus \lbrace 0 \rbrace$ and let $\widetilde \bx = [\widetilde x_0 : \dots : \widetilde x_n] \in \C P^n$. For simplicity assume that $\mathsf{Z}(\widetilde \bx) = m \widetilde Z_0 + \sum_{i = 1}^{\z - m} \widetilde Z_i$ for some $m \geq 1$, with $\widetilde Z_j \not = \widetilde Z_k$ for $j \not = k$ and $\widetilde Z_j \not \in \mathbf{R} \cup \mathbf{S} \cup \lbrace Q \rbrace$ for $j \geq 1$ (other cases work similarly and require multiple charts). Choose a sufficiently small local chart $(\zeta, V)$ on $\mathcal R$ such that $\widetilde Z_0\in V$ but $\widetilde Z_j \not \in V$ for $j \geq 1$. Let us then choose a sufficiently small open set $X \subset \C^{n+1}$ containing $(\widetilde x_0, \dots, \widetilde x_n)$ such that for all $\vec{x} = (x_0, \dots, x_n) \in X$ the element $\bx = [x_0 : \dots : x_n]$ is well-defined in $\C P ^n$ and there are exactly $m$ zeros of $\mathsf{Z}(\bx)$ denoted by $Z_{0,1}, \dots, Z_{0,m}$ in  $V$ (which can coincide), while all other zeros remain outside $V$, remain pairwise distinct and remain disjoint from $\mathbf{R} \cup \mathbf{S}\cup \lbrace Q \rbrace$. Finally let us assume that w.l.o.g.~$\lbrace R_1, \dots, R_k \rbrace = \mathbf{R} \cap V$, $\lbrace S_1, \dots, S_\ell\rbrace = \mathbf{S} \cap V$ and that $Q \in V$ (recall $\vec{q} = u(Q)$). Then it follows from the properties of $\vartheta_{\ve}$, in particular its zero structure, and the definitions of $\Pi_1$, $\Pi_2$ (cf.~\eqref{Pipropto}, \eqref{Pi2Propto}) that
    \begin{align*}
        \ff(\bx; \Sigma, \, \vec{e}, \, \vec{q} \,) &= h(x_0, \dots, x_n) \times \prod_{i,j \in \lbrace{1, \dots, m\rbrace}, i \not = j} (\zeta(Z_{0,i})-\zeta(Z_{0,j}))
        \\
        &\times \prod_{j = 1}^k \Bigg[ \prod_{i = 1}^m (\zeta(Z_{0,i})-\zeta(R_j)) \Bigg]^{-\z -1} \times \prod_{j = 1}^\ell \Bigg[ \prod_{i = 1}^m (\zeta(Z_{0,i})-\zeta(S_j)) \Bigg]^{\z -1}
        \\
        &\times \Bigg[\prod_{i = 1}^m (\zeta(Z_{0,i})-\zeta(Q))\Bigg]^{-2 \, \cdot \, \z},
    \end{align*}
    where $h \colon X \to \C$ is a nonvanishing holomorphic function. We know by Lemma \ref{LemG} that
    \begin{align*}
        &\frac{\prod_{i = 1}^m (\zeta(Z_{0,i})-\zeta(R_j))}{L^{R_j}(x_0,  \dots, x_n)}, \qquad \frac{\prod_{i = 1}^m (\zeta(Z_{0,i})-\zeta(S_j))}{L^{S_j}(x_0, \dots, x_n)},
        \qquad \frac{\prod_{i = 1}^m (\zeta(Z_{0,i})-\zeta(Q))}{L^{Q}(x_0, \dots, x_n)}
    \end{align*}
    define nonvanishing holomorphic functions on $X$. Moreover, by assumption we have that $L^{R_\alpha}(\vec{x})$ and $L^{S_\beta}(\vec{x})$ will not vanish for $\vec{x} \in X$ and $\alpha > k$, $\beta > \ell$. It follows that
    \begin{align*}
        \Delta_{\boldsymbol{\eta}}(x_0, \dots, x_n) = \widetilde h(x_0, \dots, x_n) \prod_{i,j \in \lbrace{1, \dots, m\rbrace}, i \not = j} (\zeta(Z_{0,i})-\zeta(Z_{0,j}))
    \end{align*}
    where $\widetilde h \colon X \to \C$ is a nonvanishing holomorphic function. This shows that $\Delta_{\boldsymbol{\eta}}$ is locally, hence on all of $\C^{n+1} \setminus \lbrace 0 \rbrace$, a holomorphic function. It is the $\be$-discriminant as in the setting above we have for  $\vec{x} \in X$ that
    \begin{align*}
        \prod_{i,j \in \lbrace{1, \dots, m\rbrace}, i \not = j} (\zeta(Z_{0,i})-\zeta(Z_{0,j})) \propto \prod_{i,j = 1, i \not = j}^{\# z} \dist(Z_i, Z_j)
    \end{align*}
    where $Z_i, Z_j$ on the right-hand side denote all the zeros in $\mathsf{Z}(\bx)$ (note that only the zeros $Z_{0,1}, \dots, Z_{0,m}$ can coalesce for $\vec{x} \in X$ due to our assumptions on $X$).
    It now follows from the definition \eqref{DiscrDef} that $\Delta_{\boldsymbol{\eta}}$ is a homogeneous \emph{polynomial} and one can compute directly that 
    \begin{align*}
        \Delta_{\boldsymbol{\eta}}(\lambda x_0, \dots, \lambda x_n) = \lambda^{2g-2+2 \cdot \z} \Delta_{\boldsymbol{\eta}}(x_0, \dots, x_n), \qquad \lambda \in \C \setminus \lbrace 0 \rbrace.
    \end{align*}
    This finishes the proof.
\end{proof}
\begin{remark}
    It is clear that as stated $\Delta_{\boldsymbol{\eta}}$ is only defined up to a nonzero multiplicative constant, which is not significant for our purposes. 
\end{remark}
\begin{remark}
    Note that while formula \eqref{DiscrDef} gives us the degree of $\Delta_{\be}$, it does not provide us with its coefficients. In fact, computing the coefficients of discriminants/resultants is a nontrivial problem even in the classical case \cite[Ch.~12]{GKZ94}. We will not attempt to obtain, either exact or numerical, expressions for the coefficients of $\Delta_{\be}$ in the present paper.
\end{remark}
\subsection{Generalizations}\label{Sect:Generealizations} Note that in the construction of the $\be$-discriminant only two properties of the zero map $\mathsf{Z}(\bx)$ were essential:
\begin{enumerate}[(i)]
    \item Global property: $\mathsf{Z}(\bx)$ is mapped under the Abel map to a constant independent of $\bx \in \C P^n$.
    \item Local property: $\mathsf{Z}(\bx)$ is continuous in $\bx \in \C P^n$ and the individual $Z_j$, if distinct, are locally analytic in $\bx \in \C P^n$.
\end{enumerate}
Instead of $\be$, one could view some zero map $\mathsf{Z} \colon \C P^n \to \Sym^{\#z} \mathcal R$ satisfying the above two properties as the fundamental object and define a corresponding $\mathsf{Z}$-discriminant. A simple example would be to consider the map
\begin{align*}
    \widetilde{\mathsf{Z}} \colon \bx \mapsto \mathsf{Z}(\bx) + \mathsf D,
\end{align*}
where $\mathsf{Z}$ is defined for some $\be(\vec{x})$ as in \eqref{DefZ} and $\mathsf D$ is some positive divisor. In other words, we can always add stationary zeros provided the stationary zeros of the resulting $\widetilde{\mathsf{Z}}$ will not be of higher multiplicity; the above two conditions will remain valid and the degree of the $\widetilde{\mathsf{Z}}$-discriminant will be $2g-2+2 \cdot \# \widetilde z$ with $\# \widetilde z = \z + \deg(\mathsf D)$.  

Alternatively, if $\mathsf S$ is a subset of the stationary zeros of $\be(\vec{x})$, we could instead define
\begin{align}\label{Zhat}
    \widehat{\mathsf{Z}}\colon \bx \mapsto  \mathsf{Z}(\bx) - \mathsf S.
\end{align}
Again, the construction of the corresponding $\widehat{\mathsf{Z}}$-discriminant would mirror the construction of the $\be$-discriminant. Interestingly, this case appears naturally in dimer models of (quasi-)periodic tilings of the hexagon in the ramified case, see Section \ref{SectHex}.
\begin{remark}\label{Rmk:Bundles}
    Alternatively, adding and subtracting stationary zeros (or poles) can be achieved by tensoring the holomorphic line bundle $\mathcal L \to \mathcal R$ with an auxiliary holomorphic line bundle $[\mathsf C] \to \mathcal R$ given in terms of a divisor $\mathsf C$ on $\mathcal R$. Here, $[\mathsf{C}]$ denotes the line bundle which has a meromorphic section $s$ with $\emph{Divisor}(s) = \mathsf{C}$,  see \cite[Ch.~1]{GH94}. If $s$ is such a meromorphic section  (which always exists by \cite[p.~135]{GH94}), it follows that $\widetilde \be = s \be$ is a meromorphic section of $\mathcal L \otimes [\mathsf C]$ and
    \begin{align*}
        \emph{Divisor}(\widetilde \be(\vec{x})) = \mathsf C + \emph{Divisor}(\be(\vec{x})).
    \end{align*}
     In particular, one can always assume w.l.o.g that $\be$ is a \emph{holomorphic} section of some line bundle, or alternatively that $\be$ is a meromorphic \emph{function}.
\end{remark}
\section{Results on arctic curves}\label{Sect:ArcticCurves}
\subsection{Irreducibility of \texorpdfstring{$\be$}{}-discriminants related to dimer models}
Before we state our results on arctic curve we first need to address the issue of \emph{irreducibility} of $\be$-discriminants. Note that our Hypothesis \ref{HypoEta} is too general to guarantee irreducibility. In fact, consider the simple case of $\mathcal R = \C P^1$, and $\eta_0 = 1$, $\eta_1 = z^2$ meromorphic functions in the standard chart $(z, \C)$. Then the discriminant of $\be(\vec{x}) = x_0 \eta_0 + x_1 \eta_1$ is given according to Section \ref{SectGenus0} by $\Delta_{\be}(x_0, x_1) = \Delta_{2}(x_0, 0, x_1) = - 4x_0 x_1$. Here, $\Delta_2$ is just the classical discriminant for quadratic polynomials. Clearly, in that case $\Delta_{\be}$ factorizes.

The reason irreducibility is important is because arctic curves $\AC$, unlike the \emph{complex} arctic curves $\mathcal A \mathcal H \mathcal S$ in \eqref{AHS} given by the zeros of $\be$-discriminants in a projective space, are curves in the \emph{real} plane. Hence, the degree of an arctic curve is the minimal degree of a polynomial whose real zero set coincides with the arctic curve. A priori, this implies that the degree of the $\be$-discriminant $\Delta_{\be}(x_0, x_1, x_2)$ will only be an upper bound on the degree of the corresponding arctic curve given by
\begin{align*}
\AC = \lbrace (x_1, x_2) \in \R^2 \colon \Delta_{\be}(1, x_1, x_2) = 0 \rbrace.
\end{align*}
Irreducibility of $\Delta_{\be}$ turns out to be crucial in showing equality between its degree and the degree of the corresponding arctic curve.

Let us from now assume that $\be(\vec{x}) = x_0 \eta_0 + x_1 \eta_1 + x_2 \eta_2$ where the $\eta_j$ are some meromorphic differentials on a compact Riemann surface $\mathcal R$. We will express $\eta_j$ in local charts by $f_j(\zeta) d\zeta$. The following hypothesis will play a key role in showing that $\be$-discriminants associated to dimer models are irreducible.\footnote{As we are only interested in its zero set, it will be more natural in the following hypothesis to view $\be$ as a function of $\bx \in \C P^2$ instead of $\vec{x} \in \C^{3}\setminus \lbrace 0 \rbrace$.}
\begin{hypothesis}\label{HypothesisX}
    Assume there exists a closed subset $X \subset \mathcal R$ such that
    \begin{enumerate}[(i)]
        \item for $\bx \in \R P^2$ the meromorphic differential $\be(\bx)$ has all its zeros of higher multiplicity in $X$,
        \\
        \item on $X$ we have 
        \begin{align*}
    \emph{rank}\begin{pmatrix}
        f_0(\zeta) & f_1(\zeta) & f_2(\zeta)
        \\
        f_0'(\zeta) & f_1'(\zeta) & f_2'(\zeta)
    \end{pmatrix} = 2,
\end{align*}
where $\eta_j = f_j(\zeta)d\zeta$ in some local chart (this rank condition is chart independent).
    \end{enumerate}
\end{hypothesis}
In the case of (quasi-)periodic dimer models the Riemann surface $\mathcal R$ is an M-curve and is therefore naturally equipped with real ovals $X_0, \dots, X_g$ which are homeomorphic to circles. In that case we would choose $X = \cup_{i=0}^g X_i$. Moreover, with the correct labeling we then even have the stronger condition 
\begin{align*}
\det\begin{pmatrix}
    f_1 & f_2
    \\
    f_1' & f_2'
\end{pmatrix} \not = 0
\end{align*} 
on the union $X$ of the real ovals, see \cite[Lem.~12]{BB24}. In fact, the vanishing of the above determinant would imply a coalescence of zeros for $\be(\vec{x}_{\ker})$ `at infinity', that is for 
\begin{align*}
    \vec{x}_{\ker} = \begin{pmatrix}
        0
        \\
        \widetilde x_1
        \\
        \widetilde x_2
    \end{pmatrix} \in \ker\begin{pmatrix}
    f_0 & f_1 & f_2
    \\
    f_0' & f_1' & f_2'
\end{pmatrix}
\end{align*}
with $(\widetilde x_1, \widetilde x_2) \in \R^2 \setminus \lbrace 0 \rbrace$. However, a simple counting argument shows that the zeros of $\be$ must be distinct for $\bx \in \R P^2 \setminus \Omega$ (cf.~\cite[p.~44]{BB23+}, \cite[Lem.~12]{BB24}). A similar argument holds for tilings of the hexagon studied in \cite[Sect.~7]{BB24}. This is not surprising as the arctic curve is confined to a finite region $\Omega$. 

We can know state a key technical result.
\begin{lemma}\label{LemIrr}
    Under Hypothesis \ref{HypothesisX} the $\be$-discriminant $\Delta_{\be}$ is a power of an irreducible homogeneous polynomial. 
\end{lemma}
\begin{proof}
    We first show that 
    \begin{align}\label{2x3matrix}
    \text{rank}\begin{pmatrix}
        f_0 & f_1 & f_2
        \\
        f_0' & f_1' & f_2'
    \end{pmatrix} = 2
    \end{align}
    on all of $\mathcal R$. Assume that the rank is lower for some $P \in \mathcal R$. Then by Hypothesis \ref{HypothesisX}(ii) we must have $P \not \in X$. Moreover, the rank of the matrix in \eqref{2x3matrix} at $P$ must then be equal to $1$, as otherwise $\be(\bx)$ would have for all $\bx \in \C P^2$ a higher order zero at $P$, contradicting Hypothesis \ref{HypothesisX}(i) (or even Hypothesis \ref{HypoEta}). 

    Thus, there is an up to multiples unique nonzero vector 
    \begin{align*}
        \vec{v} \in \text{span}_{\C}\big\lbrace \big(f_0(\zeta_P), f_1(\zeta_P), f_2(\zeta_P)\big)^T, \big(f_0'(\zeta_P), f_1'(\zeta_P), f_2'(\zeta_P)\big)^T\big\rbrace
    \end{align*}
    such that for all $\vec{x} \in \C^3 \setminus \lbrace 0 \rbrace$ satisfying $\langle \vec{x} , \vec{v} \rangle = 0$ (without complex conjugation) the corresponding $\be(\bx)$ will have a zero of higher order at $P$. Under the additional assumption that $\vec{x} \in \R^3$, the equation $\langle \vec{x} , \vec{v} \rangle = 0$ becomes
    \begin{align*}
        \begin{pmatrix}
            a_0 & a_1 & a_2
            \\
            b_0 & b_1 & b_2
        \end{pmatrix}\vec{x} = 0,
    \end{align*}
    where $v_j = a_j + i b_j$ with $a_j, b_j \in \R$. Clearly, this equation must have at least one nonzero solution $\vec{x}_{\text{real}} \in \R^3$. Then for the corresponding $\bx_{\text{real}} \in \R P^2$ we have that $\be(\bx_{\text{real}})$ has a higher order zero at $P \not \in X$, contradicting Hypothesis \ref{HypothesisX}(i).

    Thus we see that \eqref{2x3matrix} holds everywhere on $\mathcal R$. In particular, for each $S \in \cR$ there is a unique $\bx(S) \in \C P^2$ such that $\be(\bx(S))$ has a higher order zero at $S$. Assume now that $\Delta_{\be}$ has the following factorization into irreducible (necessarily homogeneous) polynomials:
\begin{align*}
    \Delta_{\be}(\vec{x}) = p_1(\vec{x}) p_2(\vec{x}) \dots p_m(\vec{x}).
\end{align*}
Locally on $\mathcal R$ one can always find a holomorphic mapping, e.g.\footnote{Restricting $S$ to lie on the real ovals, the mapping $S \to \vec{x}(S)$ becomes equivalent to the parametrization of the arctic curve in \cite[Rmk.~4.13]{BB23+} and \cite[Prop.~11]{BB24}.}
\begin{align*}
    S \to \vec{x}(S) = \begin{pmatrix}
        f_0(\zeta_S)
        \\
        f_1(\zeta_S)
        \\
        f_2(\zeta_S)
    \end{pmatrix} \times 
    \begin{pmatrix}
        f_0'(\zeta_S)
        \\
        f_1'(\zeta_S)
        \\
        f_2'(\zeta_S)
    \end{pmatrix} \in \C^3 \setminus \lbrace 0 \rbrace,
\end{align*}
lifting the previous mapping $S \to \bx(S) \in \C P^2$. It follows that $\Delta_{\be}(\vec{x}(S))$ vanishes locally on $\mathcal R$ (due to homogeneity the choice of the actual lift does not matter). In particular as $\mathcal R$ is connected, there is an index $k \in \lbrace 1, \dots, m \rbrace$, such that $p_k(\vec{x}(S))$ vanishes on all of $\cR$ (one might have to use multiple lifts; again their choice will not matter). Now for any zero $\vec{x}_0$ of $\Delta_{\be}$ there exists (at least one) $S \in \mathcal R$ such that $\be(\vec{x}_0)$ has a higher order zero at $S$. By the previous rank condition we must have $\vec{x}_0 = \vec{x}(S)$ (up to nonzero multiples). It follows that 
\begin{align*}
    \Delta_{\be}(\vec{x}_0) = 0 \quad \Longrightarrow \quad p_k(\vec{x}_0) = 0.
\end{align*}
Note the other implication holds trivially. From Hilbert's Nullstellensatz (see e.g.~\cite[Thm.~1.3A]{Hartshorne}) we can conclude that there exists an exponent $r \in \mathbb N$ such that $p_k^r$ is divisible by $\Delta_{\be}$.  But this implies, by the unique factorization property, that up to  nonzero multiplicative constants $p_k = p_j$ for all $j \not = k$. Therefore $\Delta_{\be} = p^m$ for some irreducible homogeneous polynomial $p$. This finishes the proof.
\end{proof}
\subsection{Arctic curves of the Aztec diamond}\label{Sect:ArcAzt}
From Theorem \ref{mainTheorem} and Lemma \ref{LemIrr}
we easily obtain the algebraic equation and the degree of arctic curves of nondegenerate $k \times \ell$-periodic models of the Aztec diamond. We recall from \cite{BB23+} that such models have $g = (k-1)(\ell-1)$ many smooth regions and $\# p = 2(k + \ell)$ many frozen regions.
\begin{theorem}\label{Thm:Aztec2}
    Consider a nondegenerate $k \times \ell$-periodic model of the Aztec diamond and let $\be$ be the associated meromorphic differential $\be = dF$ from Lemma \ref{LemmaZeros}. Then the polynomial $A(x_1, x_2) = \Delta_{\be}(1, x_1, x_2)$ is of degree $6g + 2\cdot \# p - 6 = 6 k \ell - 2(k+\ell)$ and is the polynomial of minimal degree satisfying 
    \begin{align*}
        \AC = \lbrace (x_1,x_2) \in \R^2 \colon A(x_1, x_2) = 0  \rbrace,
    \end{align*}
    where $\AC \subset [-1,1]^2$ is the arctic curve of the model. In particular, formula \eqref{DegFormula} holds.
\end{theorem}
\begin{proof}
    By Theorem \ref{mainTheorem} we know that the $\be$-discriminant $\Delta_{\be}$ is of degree $2g-2 + 2 \cdot \# z = 6g + 2 \cdot  \# p-6$. Moreover, as Hypothesis \ref{HypothesisX} is satisfied for this model by choosing $X$ to be the union of the real ovals of $\mathcal R$, we know by Lemma \ref{LemIrr} that $\Delta_{\be}$ is a power of an irreducible polynomial. However, as away from the cusp points of the arctic curve there is a \emph{simple} coalescence of zeros, it already follows that $A(x_1, x_2) = \Delta_{\be}(1, x_1, x_2)$ will vanish to first order, hence $\Delta_{\be}$ is in fact irreducible. Irreducibility of $A$ follows directly from the irreducibility of $\Delta_{\be}$. Moreover, the degree of $A$ equals the degree of $\Delta_{\be}$ as otherwise one could factor out a power of $x_0$ from $\Delta_{\be}$, again contradicting irreducibility. Now it follows from Lemma \ref{LemmaZeros} that the zero set of $A$ coincides exactly with the arctic curve of the model and the statement on the minimality is implied by the irreducibility of $A$. This finishes the proof. 
\end{proof}
\begin{remark}
    Note that the argument above implies that $\AC$ is an irreducible real algebraic curve. This fact was already observed for a wide class of dimer models in \cite[Thm.~5.15]{ADPZ} using the variational characterization of the height function.
\end{remark}

\begin{remark}
    Theorem \ref{Thm:Aztec2}  gives us a very compact proof of the arctic circle theorem for uniform tilings of the Aztec diamond due to Jockusch, Propp and Shor \cite{JPS}. In fact, any bounded degree two curve with $\frac{\pi}{2}$-rotational symmetry must already be a circle.
\end{remark}
Let us now briefly mention the quasi-periodic tilings of the Aztec diamond recently studied in \cite{BB24}, \cite{BdT24} (we will mostly follow \cite{BB24}). Such quasi-periodic models are given in terms of so-called Harnack data (see also \cite[Sect.~8]{BBS24}), denoted by $\lbrace \mathcal R, \lbrace \alpha_i^\pm, \beta_j^\pm \rbrace \rbrace$, consisting of an M-curve $\mathcal R$ playing the role of the spectral curve and a choice of finitely many special points $\lbrace \alpha_i^\pm, \beta_j^\pm \rbrace = \lbrace \alpha^\pm_1, \dots, \alpha^\pm_m,  \beta^\pm_1, \dots, \beta^\pm_m \rbrace \in X_0 \subset \mathcal R$, also called \emph{angles}, satisfying a cyclic condition one of the real ovals $X_0$, cf.~Fig.~\ref{Fig:X0_quasi}. 
\begin{figure}
\vspace{23pt}
    \begin{tikzpicture}
  \coordinate (center) at (0,2);
  \def\radius{2cm}
  \draw (center) circle[radius=\radius];

  \fill[red] (center) ++(0:\radius) circle[radius=2pt];
  \fill[red] (center) ++(-20:\radius) circle[radius=2pt];
  \fill[red] (center) ++(20:\radius) circle[radius=2pt];
  \fill[red] (center) ++(180:\radius) circle[radius=2pt];
  \fill[red] (center) ++(200:\radius) circle[radius=2pt];
  \fill[red] (center) ++(160:\radius) circle[radius=2pt];


  \fill[blue] (center) ++(90:\radius) circle[radius=2pt];
  \fill[blue] (center) ++(70:\radius) circle[radius=2pt];
  \fill[blue] (center) ++(110:\radius) circle[radius=2pt];
  \fill[blue] (center) ++(-90:\radius) circle[radius=2pt];
  \fill[blue] (center) ++(-70:\radius) circle[radius=2pt];
  \fill[blue] (center) ++(-110:\radius) circle[radius=2pt];

\put(-20,-12){\textcolor{blue}{$\underbrace{\hspace{40pt}}_{}$}};
  \put(-24,-27){\textcolor{blue}{$\alpha^-_{1}, \dots, \alpha^-_{m}$}};

  \put(-20,123){\textcolor{blue}{$\overbrace{\hspace{40pt}}_{}$}};
  \put(-25,135){\textcolor{blue}{$\alpha^+_{1}, \dots, \alpha^+_{m}$}};

  \put(-80,54){\textcolor{red}{$\begin{cases}
      \\
      \\
      \\
  \end{cases}$}};

  \put(69,54){\textcolor{red}{$\begin{rcases*}
\\
\\
\\
\end{rcases*}$}};

\put(-140,54){\textcolor{red}{$\beta^-_{1}, \dots, \beta^-_{m}$}};

\put(86,54){\textcolor{red}{$\beta^+_1, \dots, \beta^+_m$}};
\end{tikzpicture}
\vspace{30pt}
    \caption{Schematic view of the placement of the angles on the oval $X_0$, cf.~Fig.~\ref{Fig:UnboundedOval}.}
    
    \label{Fig:X0_quasi}
\end{figure}
The resulting weights are in general \emph{not} periodic and given in terms of so-called \emph{train tracks} and a \emph{discrete Abel map} which is defined in terms of the angles. The initial construction of these weights for quite general planar graphs goes back to a paper of Fock \cite{Fock}. 

Qualitatively, large tilings of quasi-periodic Aztec diamonds behave similarly to the periodic ones, see \cite{BB24}, \cite{BdT24}. These models are also more general and encompass the periodic weights as a special case, see \cite[Prop.~11]{BdT24}. However, as their definition would go beyond the scope of the present paper, we will refer to \cite[Sect.~2]{BBS24}, \cite[Sect.~2]{BdT24} for a detailed description.

\begin{remark}
    In \cite{BB24} the frozen, rough and smooth regions are defined in terms of the limit shape (see \cite[Lem.~13]{BB24}).
    This differs from the standard definition based on the decay rate of dimer correlations (or the fluctuations of the height function) following Kenyon, Okounkov and Sheffield in \cite{KOS06}. However, at least in the case of quasi-periodic models of the Aztec diamond, explicit double contour formulas for the inverse of the Kasteleyn matrix were obtained in \cite[Thm.~13]{BdT24}. One would expect that these formulas can be used to show the correct correlation decay rate in each region by proving convergence to the appropriate translation invariant Gibbs measure as in \cite{BB23+}.
\end{remark}

The bottom line is that the aforementioned paper \cite{BB24} provides, among others, a characterization of arctic curves for quasi-periodic models of the Aztec diamond in terms of a coalescence of zeros condition of a meromorphic differential analogous to Lemma \ref{LemmaZeros}. The meromorphic differential has the standard form $\be(1, x_1, x_2) = \eta_0 + x_1 \eta_1 + x_2 \eta_2$, $(x_1, x_2) \in [-1,1]^2$ and is introduced in \cite[Eq.~(25)]{BB24}. It follows from \cite[Table 2]{BB24} that $\be$ has a simple pole for every \emph{distinct} element of the set $\lbrace \alpha_i^\pm, \beta_j^\pm \rbrace$ of angles. Geometrically, the number $\p$ of poles of $\be$ equals the number of frozen regions (see e.g.~Fig.~\ref{Fig:kxl} for a $2\times 2$-periodic model with eight frozen regions and one smooth region).

Now one can conclude exactly as in Theorem \ref{Thm:Aztec2} that the polynomial $A(x_1, x_2) = \Delta_{\be}(1, x_1, x_2)$ has the arctic curve $\AC$ as its real zero set and is the polynomial of minimal degree to do so. Again, the degree of $A(x_1, x_2)$ is equal to $6g+2\cdot \p-6$, where $g$ is the number of smooth regions (also equal to the genus of the spectral curve $\mathcal R$) and $\p$ is the number of frozen regions (also equal to the degree of the pole divisor of $\be(\vec{x})$), agreeing with formula \eqref{DegFormula}.

\subsection{Arctic curves of the hexagon}\label{SectHex}
For the description of quasi-periodic models of the hexagon we again refer to \cite{BB24}. This model is somewhat more complicated than the quasi-periodic Aztec diamond, as the corresponding meromorphic differential, now denoted by  $\widehat \be(\vec{x})$, lives on a double cover $\widehat{\mathcal{R}}$ of the spectral curve $\mathcal R$. This necessitates a distinction between the double cover being \emph{ramified} vs.~\emph{unramified}. Both cases are discussed in \cite[Sect.~7]{BB24}.

\medskip

\noindent
\textbf{Unramified cover.} The unramified case treated in \cite[Sect.~7.2]{BB24} leads to a double cover $\widehat{\cR}$ of genus $\widehat g = 2g-1$, where $g \geq 1$ is the genus of the spectral curve $\mathcal R$. The corresponding model of the hexagon will have $\widehat g$ smooth regions and $\p$ frozen regions, where $\p$ is equal to the number of poles of $\widehat\be$.  In particular, $\widehat\be$ has $\z = \p + 2\widehat g -2$ zeros and thus the corresponding $\widehat\be$-discriminant is of degree $6\widehat g + 2\cdot \p -6$.

As explained in \cite[p.~42]{BB24}, the arctic curve is characterized by a coalescence of zeros condition for $\widehat\be$ analogous to Lemma \ref{LemmaZeros}. Thus, as for the $k \times \ell$-periodic Aztec diamond we can determine the degree of the arctic curve:
\begin{theorem}
    For the quasi-periodic model of the hexagon with an unramified cover studied in \cite[Sect.~7.2]{BB24} we have
\begin{align*}
    \textbf{\emph{degree of arctic curve}} = 6\cdot \#\textbf{\emph{smooth regions}}+2 \cdot \# \textbf{\emph{frozen regions}} - 6.
\end{align*}
\end{theorem}
This coincides with formula \eqref{DegFormula} for the Aztec diamond. Note that $\widehat g = 2g-1 \geq 1$, meaning that for such models a smooth phase is always present. 

While the above degree formula might seem universal, this is not quite so. Below we will see that in the case of the \emph{ramified} cover the degree of the arctic curve follows a different formula.

\medskip 

\noindent
\textbf{Ramified cover.} The ramified case treated in \cite[Sect.~7.1]{BB24} leads to a double cover $\widehat{\cR}$ of genus $\widehat g = 2g$ with ramification points $B_1$, $B_2$, where $g$ is the genus of the spectral curve $\mathcal R$. The corresponding model of the hexagon will have $\widehat g$ smooth regions and $\p$ frozen regions, where $\p$ is again equal to the number of poles of $\widehat\be$, see Fig.~\ref{BB24_Fig3}. Consequently, $\widehat\be$ will have $\z = \p + 2\widehat g -2$ zeros.

Unlike the case of the Aztec diamond and the hexagon in the unramified case, $\widehat\be$ will have stationary zeros at the ramification points $B_1$, $B_2$. These do \emph{not} contribute to the coalescence of zeros condition characterizing the arctic curve.\footnote{This is not stated explicitly in \cite{BB24} but follows from the diffeomorphism between $\widehat{\mathcal R}_+^\circ$ and the rough region, see \cite[Prop.~34]{BB24}.} Thus, we define the zero map following \eqref{Zhat} to be
\begin{align*}
    \widehat{\mathsf{Z}}(\bx) = \mathsf{Z}(\bx) - B_1 - B_2,
\end{align*}
where $\mathsf{Z}$ is the zero map of $\widehat\be$, i.e.~$\widehat{\mathsf{Z}}$ represents the nonstationary zeros of $\widehat \be$. Note that the degree of $\widehat {\mathsf{Z}}$ is equal to $\# \widehat z = \z -2$. It now follows that the corresponding $\widehat {\mathsf{Z}}$-discriminant is of degree $2\widehat g - 2 + 2(\z -2) = 6\widehat g + 2\cdot \p -10$. 

Again we have obtained the degree of the arctic curve:
\begin{theorem}
For the quasi-periodic model of the hexagon with a ramified cover studied in \cite[Sect.~7.1]{BB24} we have
\begin{align}\label{HexDegree}
    \textbf{\emph{degree of arctic curve}} = 6\cdot \#\textbf{\emph{smooth regions}}+2 \cdot \# \textbf{\emph{frozen regions}} - 10.
\end{align}
\end{theorem}

\begin{figure}[ht]
    \centering
    \includegraphics[width=0.6\linewidth]{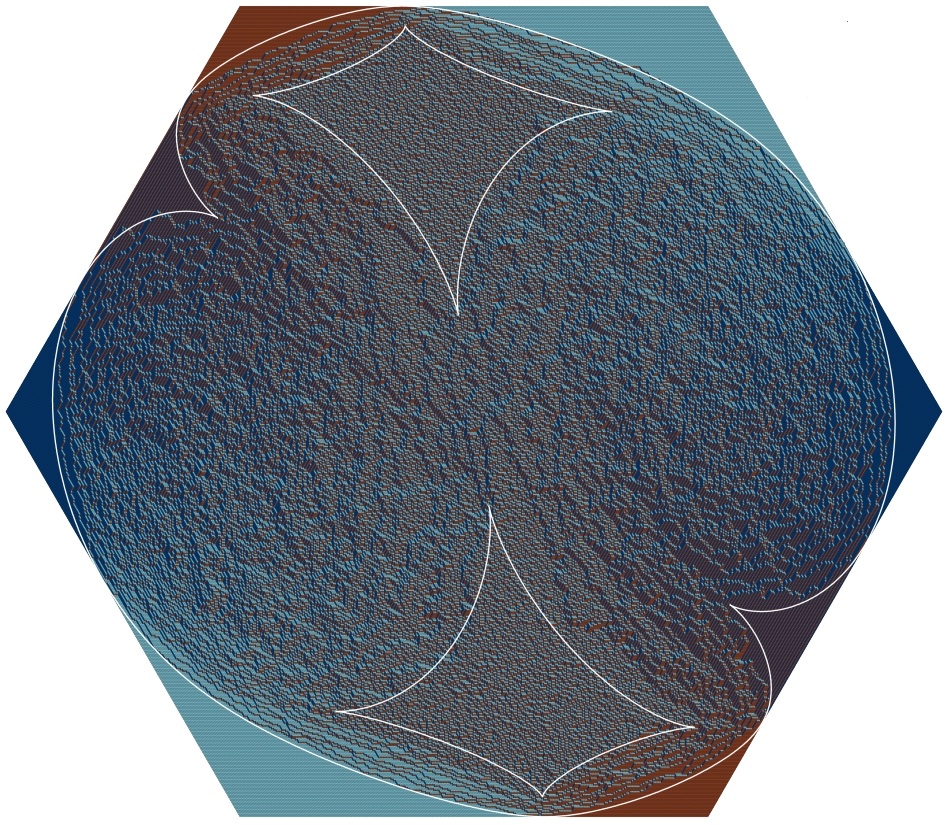}
    \caption{Large tiling of a quasi-periodic model of the hexagon in the ramified case with two smooth regions and eight frozen regions. According to formula \eqref{HexDegree} the arctic curve (in white) is of degree $6 \times 2 + 2 \times 8 -10 =18$ (image taken from \cite[Fig.~3]{BB24}).}
    \label{BB24_Fig3}
\end{figure}
\noindent
\textbf{Other models of the hexagon.} Not all quasi-periodic (or even periodic) models of the hexagon are covered by the construction outlined in \cite{BB24}. Indeed, there are other choices for the double cover $\widehat{\cR} \to \cR$ that can be made, and categorizing all possibilities seems to be a nontrivial problem. For example, in \cite{CDKL20} a one parameter family of periodic models without a smooth phase were studied. Provided the temperature parameter is sufficiently low, the double cover $\widehat{\mathcal R}$ will consist of two disjoint copies of the original spectral curve $\mathcal R$ of the model, leading to a rough region consisting of two disjoint components \cite[Prop.~2.6]{CDKL20} (cf.~Fig.~\ref{PerHex}). The arctic curve in this case is the union of two disjoint ellipses. In particular, the corresponding minimal polynomial defining the arctic curve is of degree $4$ and factorizes into two quadratic polynomials.

The two types of models of the hexagon (ramified/unramified cover) studied in \cite{BB24} always lead to a connected rough region. 
\begin{figure}[ht]
    \centering
\begin{minipage}{.5\textwidth}
  \centering
  \includegraphics[width=0.92\linewidth]{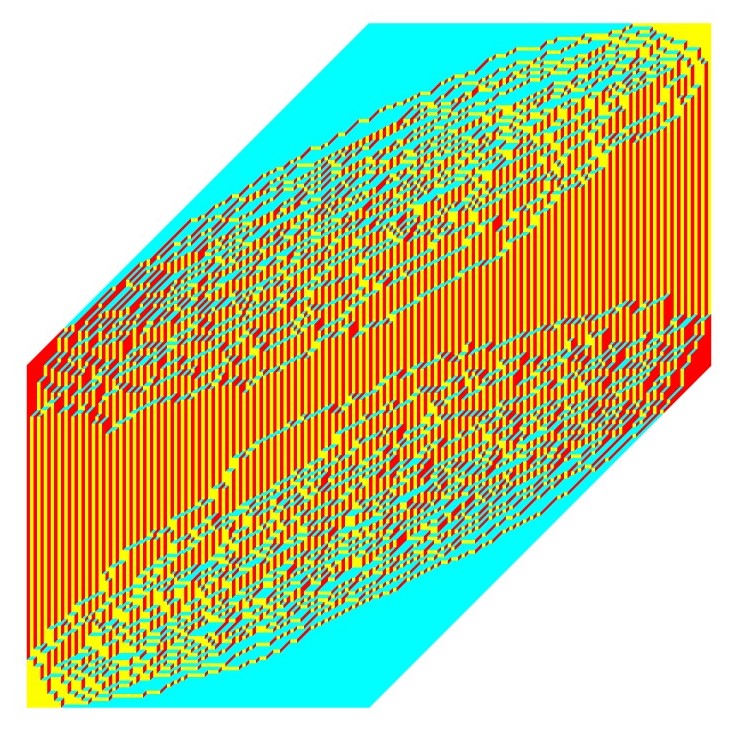}
\end{minipage}%
\begin{minipage}{.5\textwidth}
  \centering
  \includegraphics[width=0.95\linewidth]{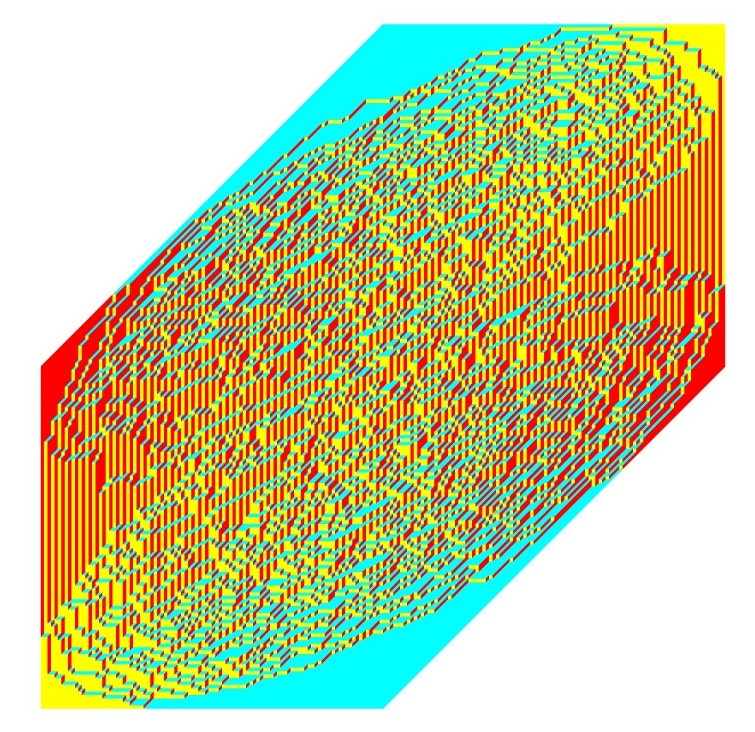}
  
\end{minipage}
\captionof{figure}{Tilings of the periodic model of the hexagon studied in \cite{CDKL20}. \textbf{Left}: Random tiling of size $100$ in the low-temperature regime $(\alpha = 0.1)$. The rough region is disconnected and the arctic curve consists of two ellipses. \textbf{Right}: Random tiling of size $100$ in the high-temperature regime $(\alpha = 0.2)$. The rough region is connected and the arctic curve is of degree $6$ (image kindly provided by Wenkui Liu).}
  \label{PerHex}
\end{figure}
Note that reducibility of the arctic curve does not contradict Lemma \ref{LemIrr}, as by definition Riemann surfaces are connected.
Hence in the low temperature regime, $\widehat \cR$ is strictly speaking not a Riemann surface as it is not connected. 

\section{Resultants on compact Riemann surfaces}\label{Sect:MixRes}
We finish this paper with a short discussion of \emph{resultants} defined on Riemann surfaces. This section is not directly related to the preceding results on arctic curves. However, it might be of separate interest.

General mixed resultants have been studied by Gelfand, Kapranov and Zelevinsky in \cite[Ch.~3.3, 8]{GKZ94}. In the following we provide explicit formulas in terms of Riemann theta functions for the resultant of two meromorphic sections of the form 
\begin{align}\label{bebnu}
\begin{split}
    \be(\vec{x}) &= x_0 \eta_0 + \dots + x_n \eta_n, \qquad \ \vec{x} \in \C^{n+1} \setminus \lbrace 0 \rbrace,
    \\
    \bnu(\vec{y}) &= y_0 \nu_0 + \dots + y_m \nu_m, \qquad\vec{y} \in \C^{m+1} \setminus \lbrace 0 \rbrace.
    \end{split}
\end{align}
Instead of Hypothesis \ref{HypoEta} we now require the set $\lbrace \eta_0, \dots, \eta_n \rbrace$ (resp.~$\lbrace \nu_0, \dots, \nu_m \rbrace$) to be linearly independent, and $\be$, $\bnu$ not having stationary zeros in common. 

We will use the terminology \emph{$(\be, \bnu)$-resultant} for the corresponding resultant. Below we show that the $\be$-discriminant can be obtained as a special case of such a resultant as in the familiar genus $0$ case, see Proposition \ref{DeltaRes}. An alternative notion of a resultant of two meromorphic functions on a Riemann surface,  characterizing the coalescence of two zeros or two poles,  was studied in \cite{GT08}.
\subsection{Construction of \texorpdfstring{$(\be, \bnu)$}{}-resultants}
To get started consider the two meromorphic sections in \eqref{bebnu}. We emphasize that the underlying holomorphic line bundles over $\mathcal R$ can be different. Let us denote by $\mathsf{Z}(\bx) = \sum_{j=1}^{\z} Z_j$ and $\mathsf{W}(\mathbf{y}) = \sum_{k=1}^{\# w} W_k$ the respective zero maps induced by $\be$ and $\bnu$. In the following we will choose some representatives $\vz_j = u(Z_j), \vec{w}_k = u(W_k) \in \C^g$, where $u$ is the usual Abel map.

Our goal is to explicitly construct a resultant-type expression for the pair $(\be, \bnu)$. That is, we are looking for a polynomial
\begin{align*}
\text{Resultant}_{\be, \bnu} \colon  \C^{n+m+2} \to \C,
\end{align*}
which is homogeneous of degrees $d_x$, $d_y$ in the variables $x_0, \dots, x_n$ and $y_0, \dots, y_m$ respectively, such that
\begin{align*}
\Big|\frac{\text{Resultant}_{\be, \bnu}(x_0, \dots, x_n\, ; \, y_0, \dots, y_m)}{|(x_0, \dots, x_n)|^{d_x}\cdot |(y_0, \dots, y_m)|^{d_y}}\Big| \propto \prod_{j, k} \dist(Z_j, W_k).
\end{align*}
To this end we define in analogy to \eqref{Pi1}  the following product:
\begin{align}\label{tPi1}
    \widetilde{\Pi}_1(\vz_1, \dots, \vz_{\z} \, ; \, \vec{w}_1, \dots, \vec{w}_{\# w}) = \prod_{j, k} \vte(\vz_j, \vec{w}_k) \in \C.
\end{align}
Here $\vte(\vr, \vs) = \theta\big(\vs - \vr + \ve\big)$ with $\vr, \vs \in \C^g$, and $\ve \in \C^g$ is chosen such that the function $\Ee(R, S) = \theta\Big(\int_R^S \vo + \ve\Big)$ with $R, \, S \in \cR$ does not vanish identically if either $R$ or $S$ is one of the stationary zeros of $\be$ or $\bnu$ respectively. Furthermore, as in \eqref{SigmaRule}, we assume that the representatives $\vz_j$, $\vec{w}_k$ are always chosen such that 
\begin{align}\label{SigmaRule2}
    \vz_1 + \dots + \vz_{\z} = \Sigma_{\be}, \quad \vec{w}_1 + \dots + \vec{w}_{\# w} = \Sigma_{\bnu},
\end{align}
for some appropriate vectors $\Sigma_{\be}, \Sigma_{\bnu} \in \C^g$.

As we have already seen for $\Pi_1$ defined in \eqref{Pi1}, the expression \eqref{tPi1} will depend on the choices of $\vec{z}_j$, $\vec{w}_k$ even after taking into account \eqref{SigmaRule2}. To correct for this we define similarly to \eqref{Pi2} two additional products given by
\begin{align*}
    \widetilde \Pi_{2,z}(\vz_1, \dots, \vz_{\z} \, ; \, \vec{q} \,) = \prod_{j=1}^{\z} \vte(\vz_j, \vq \, ), \qquad \widetilde \Pi_{2,w}(\vec{w}_1, \dots, \vec{w}_{\# w} \, ; \, \vec{p}) = \prod_{k=1}^{\vec{w}} \vte(\vec{p}, \vec{w}_k).
\end{align*} Here we assume that $\vq = u(Q)$, $\vec{p} = u(P)$, where $Q,P \in \cR$ are choose such that $\Ee(\,\cdot\, , Q)$, $\Ee(P, \,\cdot\,)$ are not identically zero.

One can easily check using \eqref{Pi2Mono} that the expression $\RR$ given by
\begin{align*}
    \RR(\bx, \mathbf{y}) =\frac{\widetilde\Pi_1(\vz_1, \dots, \vz_{\z}\, ; \, \vec{w}_1, \dots, \vec{w}_{\# w})}{[\widetilde \Pi_{2,z}(\vz_1, \dots, \vz_{\z}\, ; \, \vq \, )]^{\# w}[\widetilde \Pi_{2,w}(\vec{w}_1, \dots, \vec{w}_{\# w}\, ; \, \vec{p} \, )]^{\z}}
\end{align*}
will not depend on the choice of representatives $\vz_j$, $\vec{w}_k$ provided \eqref{SigmaRule2} is satisfied. In fact, $\RR$ is a meromorphic function on $\C P^n \times \C P^m$, cf.~Proposition \ref{Prop}. For the sake of simplicity, we did not make the dependence of $\RR$ on $\Sigma_{\be}, \Sigma_{\bnu}, \vec{e}, \vq, \vec{p}$ explicit. 

As in Section \ref{Sect:Construction} we need to understand the poles and spurious zeros of $\RR$. This in done exactly as for the construction of  $\be$-discriminats, hence we only state the end result, cf.~Theorem \ref{mainTheorem}. 
\begin{theorem}\label{ThmRes}
    With the assumptions on $\be$, $\bnu$ stated above the $(\be, \bnu)$-resultant  exists and is unique up to nonzero multiples. It is explicitly given by 
    \begin{align}\label{GenRes}
    \emph{Resultant}_{\be, \bnu}(\vec{x} \, ; \, \vec{y}) &= \RR(\bx,  \mathbf{y})  \big(L^Q_{\be}(\vec{x}) \big)^{\# w} \big( L^P_{\bnu}(\vec{y}) \big)^{\# z}.
    \end{align}
    In particular, the $(\be, \bnu)$-resultant is homogeneous of degree $\# w$ in $x_0, \dots, x_n$ and homogeneous of degree $\z$ in $y_0, \dots, y_m$.
\end{theorem}
It is apparent from \eqref{GenRes} that, unlike in the case of the $\be$-discriminants, the spurious zeros are actually not present. One can easily check that they essentially cancel with some of the zeros of the denominator of $\RR$. In fact, the only remaining poles of $\RR$ with multiplicities $\# w$ and $\z$ come from the points $Q$ and $P$ respectively. 
\subsection{\texorpdfstring{$\be$}{}-discriminants via \texorpdfstring{$(\be, \bnu)$}{}-resultants} As in the classical case, one would expect that an $\be$-discriminant can be written as a special case of a $(\be, \bnu)$-resultant. For this we will assume that the $\eta_0, \dots, \eta_n$ are linearly independent meromorphic \emph{functions} which do not have a common zero.\footnote{If the $\eta_j$ are meromorphic differentials, consider instead the meromorphic functions $\widetilde \eta_j = \frac{\eta_j}{\omega}$, where $\omega$ is a nonzero holomorphic differential. Then only the pole divisor changes, while the zero sets of $\be$ and $\widetilde \be$, hence the corresponding discriminants, remain the same. 
 See also Remark \ref{Rmk:Bundles}.} Additionally we need to assume that the meromorphic differentials $d\eta_0, \dots, d\eta_n$ are linearly independent as well.

Let us assume that $\mathcal R$ is of genus $g$ and that $\be$ has $\p$ poles and therefore $\z = \p$ zeros (as $\be$ is a meromorphic function). Then $\mathbf{d}\be$ will be a meromorphic section of the canonical line bundle and will have $\p + \# d$ poles, where $\# d$ is the number of \emph{distinct} poles of the pole divisor $\mathsf{P} = \sum_{\ell=1}^{\# d} r_{\ell} P_\ell$ of $\be$ (here $r_1 + \dots + r_{\# d} = \# p$). It follows that $\mathbf{d}\be$ will have $\# w = \p + \# d + 2g-2$ zeros.

We can now apply Theorem \ref{ThmRes} to conclude that the $(\be, \mathbf{d} \be)$-resultant will be a homogeneous polynomial of degree $\z + \# w = 2g-2 + 2\cdot \z + \# d$. This is in general not the same as the degree of the $\be$-discriminant $\Delta_{\be}$, see Theorem \ref{mainTheorem}. The discrepancy stems from the presence of poles and also arises in the classical genus $g = 0$ case, see \eqref{ClassicalDiscRes}. Indeed, away from the pole divisor of $\be$, the coalescence of zeros of $\be$ is equivalent to the simultaneous vanishing of $\be$ and $\mathbf{d}\be$ as one would expect. However, if a zero coincides with a pole of $\be$, something different happens. Recall that with our convention if for some $\bx' \in \C P^n$ the order of a pole of $\be(\vec{x})$ drops from its usual value, then this point will count as a zero of $\be$ (even if $\be(\vec{x}\,')$ does not actually vanish there). However, the order of the pole of $\mathbf{d} \be(\vec{x})$
 would then also drop at $\bx'$. Hence, if $\be$ has a ``simple zero" at an element of its pole divisor $\mathsf{P}$, $\mathbf{d}\be$ will also have a ``zero" at that point, implying that the resultant will vanish, but not the discriminant. In other words, the $(\be, \mathbf{d}\be)$-resultant will have additional zeros corresponding to the coalescence of zeros and poles. This can be seen in the genus $g = 0$ case from the identity (see \cite[Ch.~12, Eq.~(1.29)]{GKZ94})
 \begin{align}\label{ClassicalDiscRes}
     \Delta_n(\mathcal P) = \frac{1}{x_n} \text{Resultant}_{n,n-1}(\mathcal P, \mathcal P'), \qquad \mathcal P(z) = \sum_{k=0}^n x_k z^k,
 \end{align}
 where $\Delta_n$ is the classical discriminant of a polynomial of degree $n$ and $\text{Resultant}_{n,n-1}$ is the classical resultant of polynomials of degree $n, n-1$ (see \cite[Ch.~12]{GKZ94} for their basic properties). Note that the right-hand side of \eqref{ClassicalDiscRes} contains a division by the leading coefficient $x_n$ correcting for the aforementioned coalescence between zeros and infinity ($=$ only pole of $\mathcal P)$. It turns out that this correction also works in the general case as we will see in the following.
 
 First, let us denote as above by $P_1, \dots, P_{\# d}$ the distinct poles of $\be$. For each pole $P_\ell$ of degree $r_\ell$ choose a chart mapping $P_\ell$ to $0$ in which $\be(\vec{x})$ has the form
 \begin{align*}
     \be(\vec{x}) = \frac{L^{P_\ell}(\vec{x})}{\zeta^{r_\ell}} + O(\zeta^{-r_\ell+1})
 \end{align*}
 (this is consistent with definition \eqref{defL}, see Remark \ref{RemarkL}). Note that $L^{P_\ell}(\vec{x})$ is a homogeneous linear function of $x_0, \dots, x_n$ and is, up to a multiplicative constant, chart independent. With this in mind, the following identity between the $\be$-discriminant and $(\be, \mathbf{d}\be)$-resultant holds. \begin{proposition}\label{DeltaRes}
      With the notation from above we have (up to multiplicative constants)
     \begin{align}\label{Delta=Res}
         \Delta_{\be}(x_0, \dots, x_n) = \frac{\emph{Resultant}_{\be, \mathbf{d}\be}(x_0, \dots, x_n; x_0, \dots, x_n)}{\prod_{\ell =1}^{\# d} L^{P_\ell}(x_0, \dots, x_n)}. 
     \end{align}
 \end{proposition}
 \begin{proof}
     Note that both side of \eqref{Delta=Res} are homogeneous of degree $2g-2+2 \cdot \z$, hence their ratio is a well-defined meromorphic function of $\C P^n$. It remains to show that this meromorphic function has no poles or zeros, and therefore is a constant.

     Because of Lemma \ref{LemG} and Theorem \ref{mainTheorem} it suffices to show that
     \begin{align}\label{LastEq}
         \underbrace{\prod_{j, k} \dist(Z_j, W_k)}_{\propto\ |\text{Resultant}_{\be, \mathbf{d}\be}(x_0, \dots, x_n; x_0, \dots, x_n)|} \propto \quad  \underbrace{ \prod_{i \not = j} \text{dist}(Z_i, Z_j)}_{\propto \ |\Delta_{\be}|} \times \underbrace{\prod_{j, \ell} \text{dist}(Z_j, P_\ell)}_{\propto \ |\prod_\ell L^{P_\ell}|}.
     \end{align}
     Recall that here the $Z_j$ are the zeros of $\be$, the $W_k$ are the zeros of $\mathbf{d}\be$ and $P_\ell$ are the distinct poles of $\be$. That both sides of \eqref{LastEq} vanish simultaneously to the same order can be shown via local analysis using charts as in the proof of Lemma \ref{LemG}. We leave the details to the reader.
 \end{proof}
Note that formula \eqref{Delta=Res} generalizes the classical identity \eqref{ClassicalDiscRes}, as in this case we have $L^\infty(\vec{x}) = x_n$.

\medskip
\noindent{\bf Acknowledgments.} The author is thankful to Arno Kuijlaars and Tomas Berggren for many helpful discussions and to Alexander I.~Bobenko, Nikolai Bobenko, C\'edric Boutillier and B\'eatrice de Tili\`ere for clarifying their work. This research was
supported by the Methusalem grant METH/21/03 – long term structural funding of
the Flemish Government, and the starting grant from the Ragnar S\"oderbergs Foundation. The author acknowledges the support of the Royal Swedish Academy of Sciences through the Mittag--Leffler Institute where parts of this work were written while the author attended the special program on Random Matrices and Scaling Limits in the Fall of 2024.

\end{document}